\documentclass[journal]{IEEEtran}
\usepackage{amsmath}
\usepackage{amssymb}
\usepackage{amsfonts}
\usepackage{graphicx}
\usepackage{subfigure}
\usepackage{stfloats}
\usepackage{epsfig}
\usepackage{epstopdf}
\usepackage{url}
\usepackage{bigstrut} 
\usepackage{multirow} 
\usepackage{array}
\usepackage{cite}
\newcommand{\real}{\mathcal{R}\text{e}}

\newcolumntype{R}{>{\centering\arraybackslash}p{3em}}
\newcolumntype{o}{>{\centering\arraybackslash}p{4em}}
\newcolumntype{d}{>{\centering\arraybackslash}p{5em}}
\newcolumntype{u}{>{\centering\arraybackslash}p{5.7em}}
\newcolumntype{S}{>{\centering\arraybackslash}p{7em}}
\newcolumntype{F}{>{\centering\arraybackslash}p{15em}}
\title{Node Placement and Distributed Magnetic Beamforming Optimization for Wireless Power Transfer 
\thanks{M. R. Vedady Moghadam is with the Department of Electrical and Computer Engineering, National University of Singapore (e-mail: elemrvm@nus.edu.sg).}
\thanks{R. Zhang is with the Department of Electrical and Computer Engineering, National University of Singapore (e-mail: elezhang@nus.edu.sg). He is also with the Institute for Infocomm Research, A*STAR, Singapore.}
}
\author{Mohammad R. Vedady Moghadam, \textit{Member}, \textit{IEEE}, and Rui Zhang, \textit{Fellow}, \textit{IEEE} \vspace{-1.5mm}} 
\begin{document} 
\maketitle\thispagestyle{empty} 
\begin{abstract} 
In multiple-input single-output (MISO) wireless power transfer (WPT) via magnetic resonant coupling, multiple  transmitters are deployed to enhance the efficiency of power transfer to a single receiver by jointly adapting their source  currents/voltages so as to constructively combine the induced  magnetic fields at the receiver, a technique  known as \textit{magnetic beamforming}.  
In practice, since the transmitters (power chargers) are usually at fixed locations and the receiver (e.g. mobile phone)  is desired to be freely located in a target region for wireless charging, its received power can fluctuate significantly over locations even with adaptive magnetic beamforming applied.
To achieve uniform power coverage, the transmitters need to be  optimally placed in the region  such that a minimum charging power  can be achieved for the receiver regardless of its location, which  motivates this paper.   
First, we derive the optimal magnetic beamforming solution in closed-form for a distributed MISO WPT system with fixed  locations of the transmitters and receiver  to maximize the deliverable power to the receiver subject to a given sum-power constraint at all transmitters. 
By applying adaptive magnetic beamforming based on this optimal solution, we then jointly optimize the locations of all transmitters to maximize the minimum power deliverable to the receiver over a given one-dimensional (1D) region.  
Although the formulated problem is non-convex, we propose an iterative  algorithm for solving it efficiently. 
Extensive simulation results are provided which show the significant  performance gains by the proposed design with optimized transmitter locations  and magnetic beamforming as compared to other benchmark schemes with non-adaptive and/or heuristic   transmitter current allocation  and node placement.      
Last, we extend the node placement problem to the case of two-dimensional (2D)  region, and propose efficient designs for this case.  
\end{abstract}
%*************************
%*************************
%*************************
%*************************
\begin{IEEEkeywords} 
Near-field wireless power transfer, distributed wireless charging, magnetic resonant coupling,  magnetic beamforming, node placement optimization,  uniform  power coverage. 
\end{IEEEkeywords}  
%*************************
%*************************
%*************************
%*************************
\newtheorem{definition}{\underline{Definition}}[section]
\newtheorem{fact}{Fact}
\newtheorem{assumption}{Assumption}
\newtheorem{theorem}{\underline{Theorem}}[section]
\newtheorem{lemma}{\underline{Lemma}}[section]
\newtheorem{corollary}{\underline{Corollary}}[section]
\newtheorem{proposition}{\underline{Proposition}}[section]
\newtheorem{example}{\underline{Example}}[section]
\newtheorem{remark}{\underline{Remark}}[section]
\newtheorem{algorithm}{\underline{Algorithm}}[section]
\newcommand{\mv}[1]{\mbox{\boldmath{$ #1 $}}}
%*************************
%*************************
%*************************
%*************************
\vspace{-1.5mm}
\section{Introduction}
\PARstart{N}EAR-FIELD wireless power transfer (WPT) has drawn significant interests recently  due to its high efficiency for delivering power to electric loads without the need for any wire. 
Inductive coupling (IC) \cite{Bolger,Kim,Chwei-Sen} is the conventional method to realize near-field WPT for short-range applications typically within  centimeters.
The wireless power consortium (WPC) that  developed the ``Qi'' standard \cite{WPC} is the main industrial organization for commercializing  wireless charging based on IC.
Recently,  magnetic resonant coupling (MRC)  \cite{Kurs, Elisenda, Shin,YZhang1,HLi} has been applied to significantly enhance the efficiency and range of WPT  compared to IC, thus opening up a broader avenue for practical applications such as biomedical device charging  \cite{Na,QXu}, electric vehicle charging  \cite{SLi,Ibrahim}, etc.  
In MRC-enabled WPT (termed MRC-WPT), compensators each being a capacitor of variable capacity are embedded in the electric circuits of  individual power transmitters and receivers to tune their oscillating  frequencies to be the same as the operating frequency adopted by all input voltage/current sources so as to achieve resonance.  
Alternatively, resonators each of which  constitutes  a simple RLC circuit resonating at the source  frequency can be employed in close proximity of the coils of any off-resonance transmitters/receivers to help efficiently transfer power between  them. 
With MRC, the total reactive power consumption in the system is effectively reduced  due to resonance and thus high power transfer efficiency is achieved over longer distance than the conventional IC. 
The preliminary experiments in \cite{Kurs} show that  MRC enables a single transmitter to transfer $60$ watts of power wirelessly with $40\%$--$50\%$ efficiency to a single receiver at a distance of about $2$ meters.  
More recent  experiment and simulation results on the controlling, scalability analysis, and  performance characterization  of MRC-WPT systems have been  reported in the literature (see e.g. \cite{Zhong,Yin,Elisenda2}).
Formed after the merging between the alliance for wireless power (AW4P) that developed the ``Rezence'' specification and the power matters alliance (PMA), the new AirFuel alliance is currently  the main industrial organization for developing wireless charging  based on MRC \cite{Rez}. 
The Rezence specification advocates a superior charging range, the capability of charging multiple devices simultaneously, and the use of  two-way communication via e.g. Bluetooth  between the charger unit and devices for real-time charging control. 
These features make Rezence and its future extensions a promising technology for high-performance near-field  wireless charging. 

In the current Rezence specification, one transmitter with a single coil is used  in the power transmitting unit, i.e., only  the  single-input multiple-output (SIMO) MRC-WPT is considered, as shown in Fig. \ref{fig:WPTTable}(a).
Although this centralized WPT system performs well when the receivers are all sufficiently close to the transmitter, the power delivered to each receiver decays significantly as it moves away from the transmitter. 
This thus motivates distributed WPT where the single centralized transmitter coil is divided into multiple coils (i.e., separate   transmitters) each with a smaller size (coil radius), which are then  placed in different locations to cover a given target region,  as shown in Fig. \ref{fig:WPTTable}(b). 
By coordinating the  transmissions of distributed coils  via jointly allocating their source currents/voltages, in \cite{JD} it is  shown that their induced magnetic fields at one or more receivers  can be constructively combined, thus achieving  a  magnetic beamforming gain in a manner analogous to multi-antenna beamforming in  far-field wireless communication or  WPT \cite{RZhang,Dusit,SBi}.   
Besides, distributed WPT shortens the maximum distance from each receiver to its nearest transmitter(s) in a given region  compared to centralized WPT,  thus achieves more uniform charging performance  over the region. 
\begin{figure} [t!]
	\begin{center}
		\subfigure[Centralized WPT]
		{\scalebox{0.45}{\includegraphics{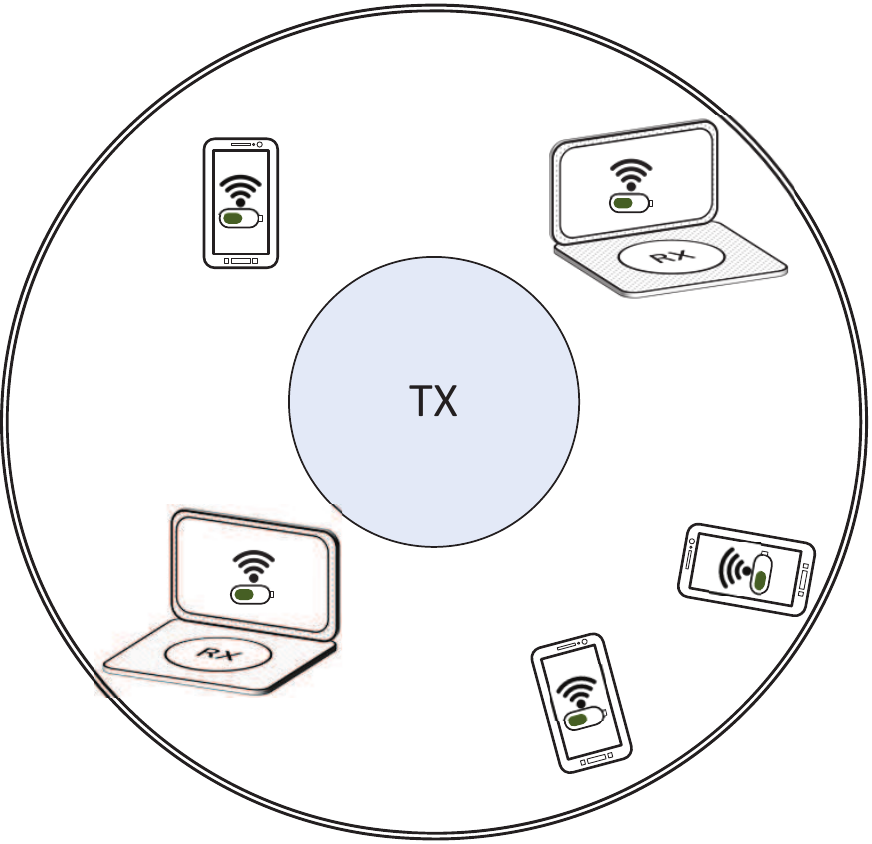}}}\hspace{4mm} 
		\subfigure[Distributed WPT]
		{\scalebox{0.45}{\includegraphics{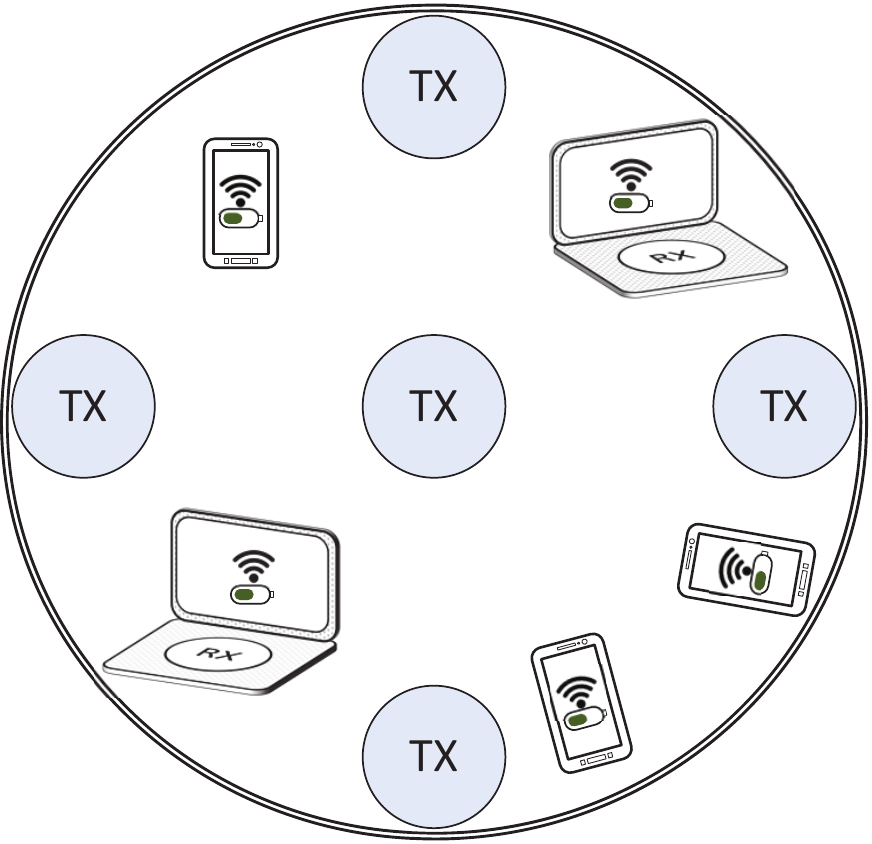}}}  \end{center} 
	\vspace{-2mm}
	\caption{Two different  system setups for wireless charging.}
	\label{fig:WPTTable}
\end{figure} 

\begin{figure} [t!]
	\begin{center} 
		\subfigure[Far-field WPT]
		{\scalebox{0.93}{\includegraphics{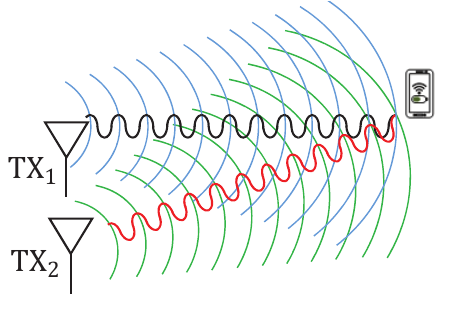}}}\hspace{-0.7mm} 
		\subfigure[Near-field WPT]
		{\scalebox{0.96}{\includegraphics{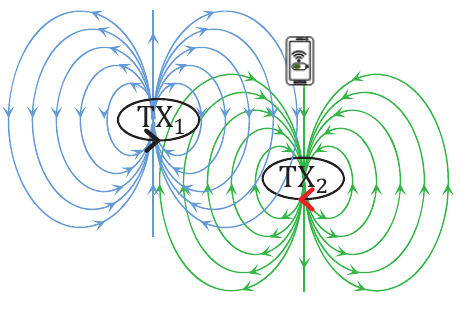}}} 
	\end{center} 
	\vspace{-2mm}
	\caption{Beamforming in far-field versus near-field MISO WPT systems.}
	\label{fig:Beamforming} \vspace{-2mm}
\end{figure} 
Generally speaking,  there are similarities as well as  differences between  magnetic beamforming in near-field WPT and its far-field counterpart. 
For instance, as shown in Fig. \ref{fig:Beamforming}(a), the objective of the beamforming design in far-field multiple-input single-output (MISO) WPT is to ensure that the electromagnetic (EM) waves propagated from different transmitter (TX) antennas are aligned at the single receiver antenna, i.e., their arriving phases are all identical so that they can be coherently added to maximize the received signal amplitude/power. 
This is practically realized by adjusting the phase and amplitude of the transmitted signal at each antenna.   
In contrast, as shown in Fig. \ref{fig:Beamforming}(b),  magnetic beamforming in near-field MISO WPT aims  to achieve that the polarities of the magnetic fields generated by different transmitters are all identical (i.e., upward or downward) at a given receiver location and hence these magnetic fields can be constructively added to maximize the magnetic flux intensity at the receiver coil.  
This can be realized by adjusting the direction and magnitude of the current flowing in each transmitter coil.    
For example, in Fig. \ref{fig:Beamforming}(b), the current at  coil $2$ should flow in the opposite direction of that in coil $1$ in order for their generated magnetic fields to be both downward at the receiver coil.   
However, different from  far-field WPT in which the receivers are passive and have no effect  on the power radiated by the transmitters, the cross-coupling effect among the transmitters and receivers in near-field WPT is generally strong. 
As a result, changing the input current/voltage (load resistance) at one transmitter (receiver) not only affects its own transmitted (received) power, but also influences the power of all other transmitters and receivers in general.  
Therefore, the modeling and design of magnetic beamforming in near-field WPT are generally different from its far-field counterpart based on EM wave propagation. 

The optimal magnetic beamforming  design in MISO MRC-WPT systems is investigated in \cite{Lang,Gang}. 
Specifically, \cite{Lang} formulates a convex optimization problem to jointly design the source currents at all transmitters to maximize the WPT efficiency subject to a given  minimum power deliverable to the load at a single receiver.  
On the other hand, \cite{Gang} jointly optimizes the transmitter currents to maximize the deliverable power to a single receiver by considering a sum-power constraint for all transmitters as well as practical peak  voltage and current  constraints at individual  transmitters.
Recently, the selective WPT technique is proposed for  MISO MRC-WPT systems where one transmitter is selected at each time (i.e., transmitter selection) to deliver wireless power to a single receiver, and  a simple control mechanism is  devised   for performance optimization \cite{BHChoi}. 
Selective WPT  is also  proposed for SIMO MRC-WPT systems  in \cite{YZhang,YJKim}.  
This technique delivers power to only one selected receiver (i.e., receiver selection) at each time to eliminate the magnetic cross-coupling  effect among different  receivers  and hence achieve more balanced power transfer to them, assuming that their natural frequencies are set  well  separated from each other. 
Alternatively, \cite{MRVM-Zhang-1}   proposes to jointly optimize the load resistance of all receivers in a SIMO MRC-WPT system  to exploit the magnetic cross-coupling  effect to alleviate the near-far issue by delivering balanced power to individual receivers regardless of their distances to the power transmitter.   
In general, selective WPT requires a simpler control mechanism than magnetic beamforming, while its performance is also limited since only one pair of transmitter and receiver is allowed for power transfer at each time. 
In contrast, magnetic beamforming enables multiple transmitters to send power to one or more receivers simultaneously by properly assigning  the source currents (load resistance values) at individual transmitters (receivers), thus in general achieving better performance than the simple  transmitter/receiver selection. 

The studies in  \cite{JD,Lang,Gang,YZhang,YJKim,MRVM-Zhang-1,BHChoi} show promising directions  to improve the efficiency as well as performance fairness in MRC-WPT systems, but all of them assume that the transmitters and receivers are at fixed locations in a target region.  
In practice, each wireless device (e.g., mobile phone) is desired to be freely located in any position in the region (e.g., on a charging table) when it is being charged, for the 
convenience of its user.   
In this case, if the locations of
the transmitters are not appropriately  designed in a MISO MRC-WPT system, the deliverable power to the receiver can fluctuate significantly over different  locations.  
Such power fluctuation degrades the quality of service, since the power requirement  of the receiver load may not be satisfied at all locations in the region, even when magnetic beamforming is applied to adapt to the location of the receiver.
To achieve uniform power coverage, one possible method is to allow the transmitters to  track the location of the receiver and move toward it in real time, which however may not be feasible as in practice transmitters like power chargers are usually at fixed locations.  Alternatively,  the transmitters can be optimally placed at their initial deployment  such that a minimum charging power is ensured to be achievable  for the receiver regardless of its location in the region.    
This thus motivates our work in this paper to optimize the transmitter locations in a MISO MRC-WPT system jointly with receiver location-adaptive  magnetic beamforming to maximize the minimum power deliverable to the receiver  over a target region.    

The main results of this paper are summarized as follows:  
\begin{itemize}
\item First, we formulate the magnetic beamforming problem for a MISO MRC-WPT system with distributed transmitters to maximize the deliverable power to the load at a single receiver subject to a given transmitters' sum-power constraint, by assuming that the power transmitters and receiver are all
at fixed locations.  
We derive the closed-form solution to the magnetic beamforming problem in terms of the mutual inductances between the transmitters and the receiver. 
Our solution shows that the optimal  current allocated to each transmitter is  proportional to the mutual inductance between its coil and that of the receiver. 
For the special case when the transmitters are sufficiently separated from each other, we show that the optimal magnetic beamforming reduces to the simple transmitter selection scheme  \cite{BHChoi} where all power is allocated to one single transmitter that has the highest  mutual inductance value with the receiver.   
\item To demonstrate the performance  gain of magnetic beamforming, we  compare it to an uncoordinated MISO WPT system with equal current allocation over all  transmitters \cite{Gang}, as well as the transmitter selection scheme \cite{BHChoi}. 
We also compare the performance of distributed WPT with magnetic beamforming versus centralized WPT subject to the same total size of transmitter coils. 
\item Next, by applying adaptive magnetic beamforming with the optimal solution derived, we formulate the node placement problem to jointly optimize the transmitter locations to maximize the minimum power deliverable to the receiver in a given one-dimensional (1D) region, i.e., a line of finite length where the receiver can be located in any point in the line. Although simplified, the 1D case is studied first  for the purpose of exposition as well as drawing useful insights.   
The formulated problem is non-convex, while  we propose an iterative  algorithm for solving it approximately by utilizing the fact that the transmitters should be symmetrically located over  the mid-point of the  target line to maximize the minimum deliverable power. 
We present extensive simulation results to verify the effectiveness of our proposed  transmitter location optimization algorithm in improving both the minimum and average deliverable power over the target line as compared to a heuristic design that uniformly locates the transmitters. 
\item At last, we extend the node placement problem to the more  general  two-dimensional (2D)  target region case, i.e., a disk in 2D with a finite radius. 
A practical scenario for this  setup could be a round table with built-in wireless chargers mounted below its surface, and a receiver that can be freely placed on its surface.
Using an example of five transmitters, we show that the design approach for the 1D case can be similarly applied to obtain the optimal locations of the transmitters under the 2D setup with magnetic beamforming to maximize the minimum deliverable power over the target disk region,  by exploiting the  property of  rotational symmetry in 2D.
\end{itemize}

The rest of this paper is organized as follows.  
Section II introduces the system model. Section III formulates the magnetic beamforming problem and presents its optimal solution. 
Section IV formulates the node placement problem for the 1D target region case, and presents an iterative algorithm for solving it. 
Section V presents simulation results  for the 1D case. 
Section VI extends the node placement problem to the 2D target region case with an example of five transmitters.  
Finally, we conclude the paper in Section VII.
%*************************
%*************************
%*************************
%*************************
\section{System Model}  \label{Sec:System-Setup} 
\vspace{-1mm}
In this paper, we consider a MISO MRC-WPT  system with $N\ge1$ identical single-coil transmitters, indexed by $n$,  $n \in \{1,\ldots,N\}$, and a single-coil receiver, indexed by $0$ for convenience. 
It is assumed that all the transmitters and receiver  are each equipped with a Bluetooth communication module to enable information exchange among them to achieve coordinated WPT \cite{Rez}.
Each transmitter $n$ is connected to a stable energy source supplying sinusoidal voltage over time given by  $\tilde{v}_n(t)=\real\{v_ne^{j w t} \}$,  where ${v}_n$ is a complex number denoting the steady state voltage in phasor form  and  $w>0$ denotes its     angular frequency. Note that $\real{\{\cdot\}}$ represents the real part of a complex number. 
On the other hand, the receiver is connected to an electric load, e.g., battery of a mobile phone.  Let $\tilde{i}_n(t)=\real\{i_ne^{j w t}\}$, where  ${i}_{n}=\overline{i}_n+j \hat{i}_n$, with $j^2=-1$, denotes the steady state current at  transmitter $n$.  
This current produces a time-varying magnetic flux in the transmitter coil, which passes through the receiver coil  and induces time-varying current in it.  
We  denote  $\tilde{i}_0(t)=\real\{i_0e^{j w t}\}$,  with ${i}_0=\overline{i}_0+j \hat{i}_0$, as the steady state current at the receiver.
\begin{figure}
	\begin{center}
		\subfigure[Case 1: $N$ is even]
		{\scalebox{0.52}{\includegraphics{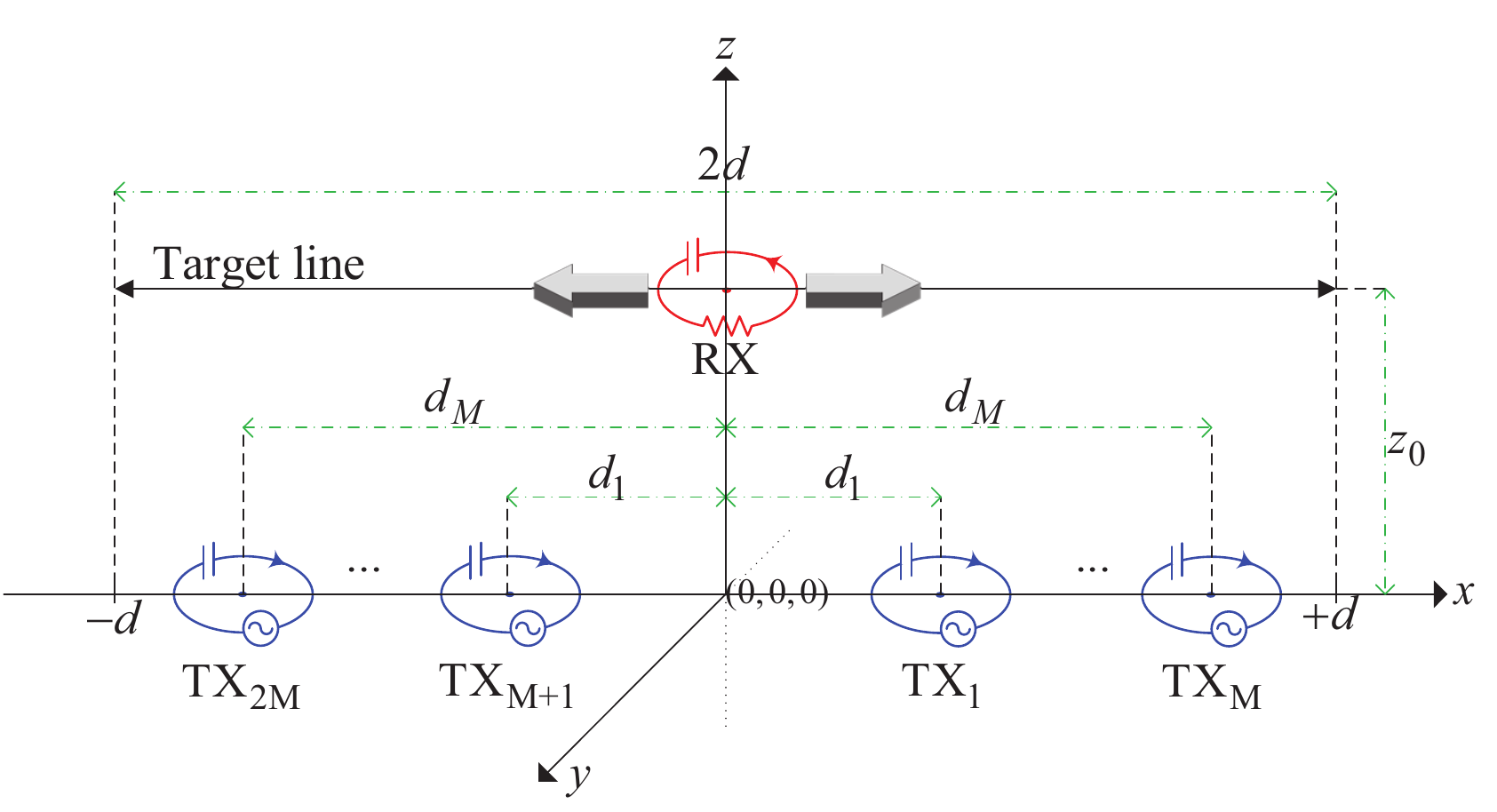}}} 
		\subfigure[Case 2: $N$ is odd]
		{\scalebox{0.52}{\includegraphics{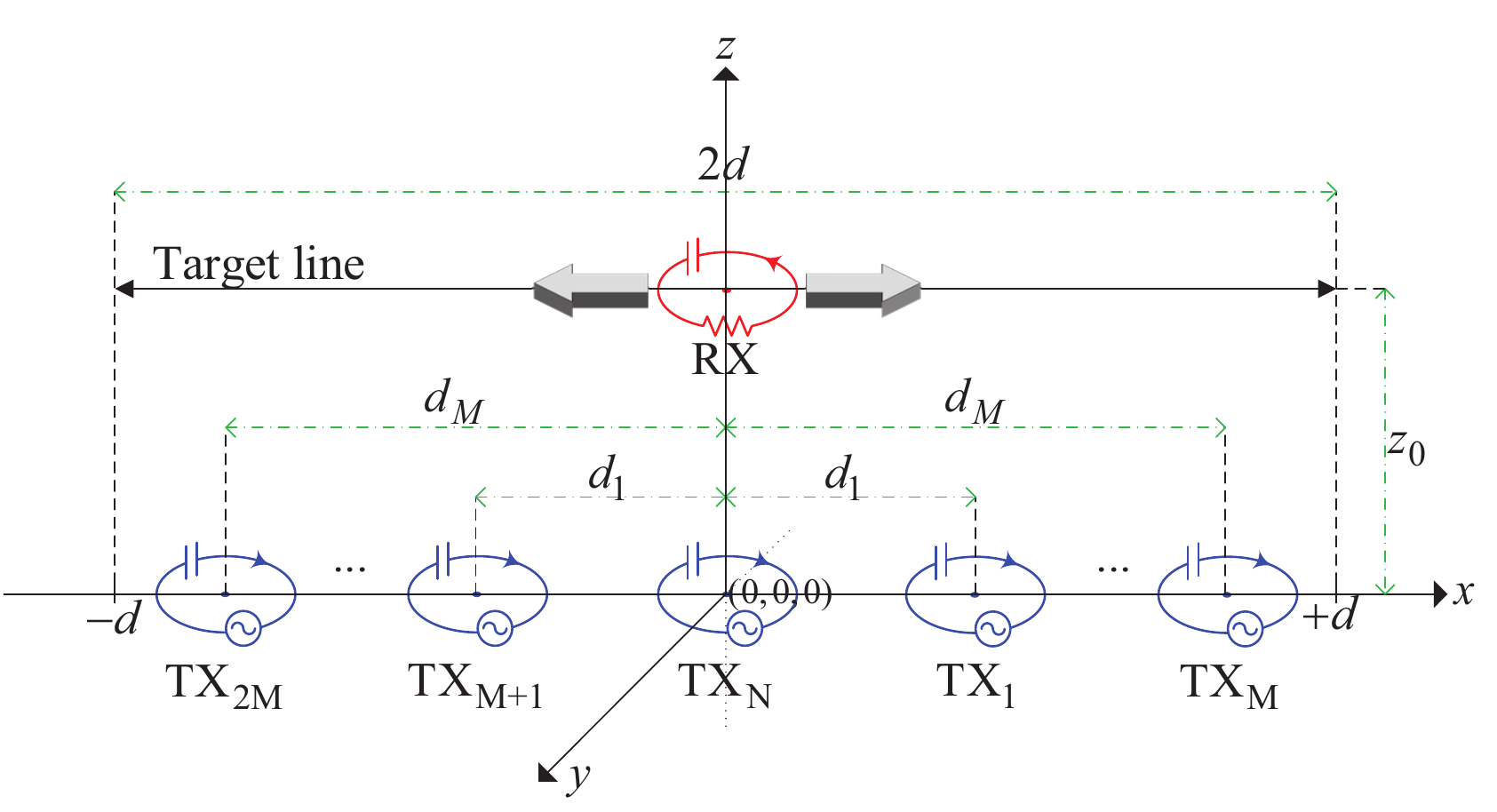}}}  
	\end{center} 
	\vspace{-2mm}
	\caption{MISO MRC-WPT system setup with 1D target line.} \label{fig:MISO-Linear} 
	\vspace{-2mm}
\end{figure} 

First, we consider the case of  WPT in 1D.  
As shown in Fig. \ref{fig:MISO-Linear}, we assume that the receiver can move horizontally along a line that lies in the $(x,y)$ plane with its ($x,y$)-coordinates satisfying  $|x| \le d$, with $d>0$,  and $y=0$ while it has a fixed height of  $z=z_0$, $z_0>0$. We denote this line as the 1D {\it target line}. 
Note that  $|\cdot|$ denotes the absolute value of a real/complex number. 
The transmitters are also installed horizontally at fixed locations along a line that is in parallel with and below  the target line, at a fixed height of $z=0$.   
Let $x_n$ with $|x_n|\le d$  ($x_0$ with $|x_0|\le d$) denote  the location of transmitter $n$  (receiver) over the $x$-axis. 
In this paper, we consider  that $x_n$'s are symmetric over $x=0$.\footnote{We will show later in Section \ref{Sec:Node Placement Optimization} that such symmetric structure of the transmitters maximizes the minimum  deliverable power over the target line.}    
Hence, we  consider the following two  cases for the symmetric deployment of the transmitters.  
\begin{itemize}
\item Case $1$: $N$ is even. In this case, let $M=N/2$ and we set  $x_{n}=-x_{M+n}=d_n$, with $0\le d_{n}\le d$,  $n=1,\ldots,M$, as shown in Fig. \ref{fig:MISO-Linear}(a).  
\item Case $2$: $N$ is odd. In this case, let $M=(N-1)/2$ and  we set  $x_{n}=-x_{M+n}=d_n$, with $0\le d_{n}\le d$, $n=1,\ldots,M$, and $x_{N}=0$, as shown in Fig. \ref{fig:MISO-Linear}(b).
\end{itemize}  

Let  $r_{\text{tx}}>0$ ($r_{\text{rx,p}}>0$), $l_{\text{tx}}>0$ ($l_{\text{rx}}>0$), and $c_{\text{tx}}>0$ ($c_{\text{rx}}>0$)  denote the  parasitic resistance,   the self-inductance, and the  capacity of the compensator in each transmitter (receiver), respectively.  Denote $r_{\text{rx,l}}>0$ as  the resistance of the load at the receiver. Accordingly, we use  $r_{\text{rx}}=r_{\text{rx,p}}+r_{\text{rx,l}}$ to represent the total ohmic resistance of the receiver. 
By assuming  that the coil  of each of the  transmitters as well as the receiver consists of multiple closely winded turns of round-shaped wire, we obtain $r_{\text{tx}}= 2 \sigma_\text{tx} b_\text{tx} e_{\text{coil},\text{tx}}/e_{\text{wire},\text{tx}}^2$ and $r_{\text{rx,p}}=2  \sigma_\text{rx} b_\text{rx} e_{\text{coil},\text{rx}}/e_{\text{wire},\text{rx}}^2$, 
where $e_{\text{coil},\text{tx}}$ ($e_{\text{coil},\text{rx}}$),  $e_{\text{wire},\text{tx}}$ ($e_{\text{wire},\text{rx}}$), $\sigma_\text{tx}$ ($\sigma_\text{rx}$), and $b_\text{tx}$  ($b_\text{rx}$)  are  the average radius of the coil of each transmitter (receiver), the radius of the wire used to make the coil, the ohmic resistivity of the wire, and the number of  turns of the coil, respectively.\footnote{In practice, multi-layer wiring technique can be used to reduce the thickness of each receiver coil  such that it can be easily fitted into a small-size  electronic device, e.g., smartphone  (see Appendix E in \cite{MRVM-Zhang-1} for more detail).}    
Furthermore, we obtain $l_\text{tx}=\mu b_{\text{tx}}^2 e_{\text{coil},\text{tx}} (\ln(8e_{\text{coil},\text{tx}}/e_{\text{wire},\text{tx}})-2)$ and $ l_\text{rx}=\mu b_{\text{rx}}^2 e_{\text{coil},\text{rx}}  (\ln( 8e_{\text{coil},\text{rx}}/e_{\text{wire},\text{rx}})-2)$, where $\mu=4 \pi \times 10^{-7}$N/A$^2$ is the magnetic permeability of the air \cite{MRVM-Zhang-1}.  
The capacities of compensators are then chosen such that the natural frequencies of the transmitters and receiver are the same as the operating  frequency adopted for the input voltage sources, i.e., $w$. 
Hence, we set $c_\text{tx}=1/(l_\text{tx} w^2)$ and $c_\text{rx}=1/(l_\text{rx} w^2)$.    
This helps compensate the reactive power consumed by the self-inductance of the coil at each transmitter/receiver, but in general reminiscent  reactive power presents due to the magnetic coupling between the transmitters and receiver.   
In practice, when the reactive power consumption in the MRC-WPT system is large, the transmitters' source voltages may  spike \cite{Tal}.  
To keep the MRC-WPT system practically feasible, either the peak voltage/current constraints  for individual transmitters need to be considered \cite{Gang}, or equivalently the complex power drawn from these sources should be minimized \cite{Kisseleff}.  
In these cases, devising the optimal magnetic beamforming solution is more complicated and hence investigating  the node placement optimization based on it would be  intractable.    
For simplicity, in this paper we only consider the  active power consumption, and hence use `power' and `active power' interchangeably, unless specified otherwise.

Let $h_{nk}$ and $h_{n0}$ be real numbers denoting the mutual inductance between the coils of transmitters $n$ and $k$, with $k\neq n$, as well as that between transmitter $n$ and the receiver, respectively.  
As shown in Fig. \ref{fig:MISO-Linear}, we consider that   the transmitters and receiver  are all located horizontally in parallel $(x,y)$ planes at heights $z=0$ and $z=z_0$, respectively,  hence their orientations are identical. In this case, from the so-called Conway's mutual inductance formula \cite{Conway},  we have\footnote{In this paper,  we assume a free space propagation model where the transmitters and receiver are all placed in an environment without   any nearby externalities absorbing and/or reflecting magnetic fields. This helps deriving  tractable solutions for the mutual inductance values given in (\ref{eq:h_nk}) and (\ref{eq:h_n0}).}           
\begin{align} 
&h_{nk}= \mu \pi b_\text{tx}^2 e_{\text{coil},\text{tx}}^2 \int_0^{\infty} J_0\left(d_{nk}u\right) \big(J_1\left(e_{\text{coil},\text{tx}} u\right) \big)^2 du, \label{eq:h_nk}\\
&h_{n0}=\mu \pi b_\text{tx} b_\text{rx} e_{\text{coil},\text{tx}} e_{\text{coil},\text{rx}} \int_0^{\infty} J_0\left(d_{n0}u\right) J_1\left(e_{\text{coil},\text{tx}} u\right)\nonumber \\
&\hspace{45.8mm}J_1\left(e_{\text{coil},\text{rx}} u\right) e^{-z_0 u}du,   \label{eq:h_n0}
\end{align}
where  $d_{nk}=|x_n-x_k|$, $d_{n0}=|x_n-x_0|$, and $J_\alpha(u)=\sum_{m=0}^{\infty} (-1)^m(u/2)^{2m+\alpha}/(m! (m+\alpha)!)$ is the  Bessel function of the first kind of order $\alpha \in \{0,1\}$  with $(\cdot)!$ denoting the factorial of a positive integer.    
The integration terms in (\ref{eq:h_nk}) and (\ref{eq:h_n0}) can be computed numerically, while there are no closed-form analytical expressions for them. 
In practice, the transmitters and  receiver commonly use small coils for WPT; therefore,   $h_{n0}$ in (\ref{eq:h_n0}) can be  simplified as follows.
\begin{lemma} \label{Lemma:mutualInductance}
If $e_{\text{coil},\text{tx}}, e_{\text{coil},\text{rx}} \ll z_0$, we have	
\begin{align} 
h_{n0} \approx \beta  \dfrac{2z_0^2-d_{n0}^2}{\sqrt{\left(z_0^2+d_{n0}^2\right)^5}}, \label{eq:hn0-simplified}
\end{align}
where $\beta =\mu \pi b_\text{tx} b_\text{rx} e_{\text{coil},\text{tx}}^2 e_{\text{coil},\text{rx}}^2/4$ is a constant with the given coil parameters.
\end{lemma} 
\begin{proof}
Please see Appendix A. 
\end{proof}

To validate the accuracy of the proposed approximation in (\ref{eq:hn0-simplified}), we consider Case $2$ in Fig. \ref{fig:MISO-Linear}(b) with $N = 5$ identical transmitters, $d = 1$m, and variable $z_0$, where the physical and electrical characteristics of the coils in the transmitters and receiver are given in Tables \ref{tab:Physical-Charac} and \ref{tab:Electrical-Charac} (see Section \ref{Sec:numerical Example}), respectively. 
We assume that the transmitters are uniformly located over $|x|\le 1$m, with $x_1 = x_3 = 0.5$m, $x_2 =  x_4 = 1$m, and $x_5 = 0$. 
Figs. \ref{fig:Mutual_Inductance_veification}(a) and \ref{fig:Mutual_Inductance_veification}(b) compare the actual and approximated values of the  mutual inductance  between transmitter $1$ and the receiver, $h_{10}$, versus the receiver's x-coordinate  $x_0$ under heights of $z_0=0.2$m and $z_0=0.4$m, respectively.  
It is observed that the approximation is  tight in general; whereas there are  discrepancies at $x_0=0$. 
It is also observed that the  discrepancies decrease  when $z_0$ increases. 
\begin{figure}
	\begin{center}
		\subfigure[$z_0=0.2$m]
		{\scalebox{0.52}{\includegraphics{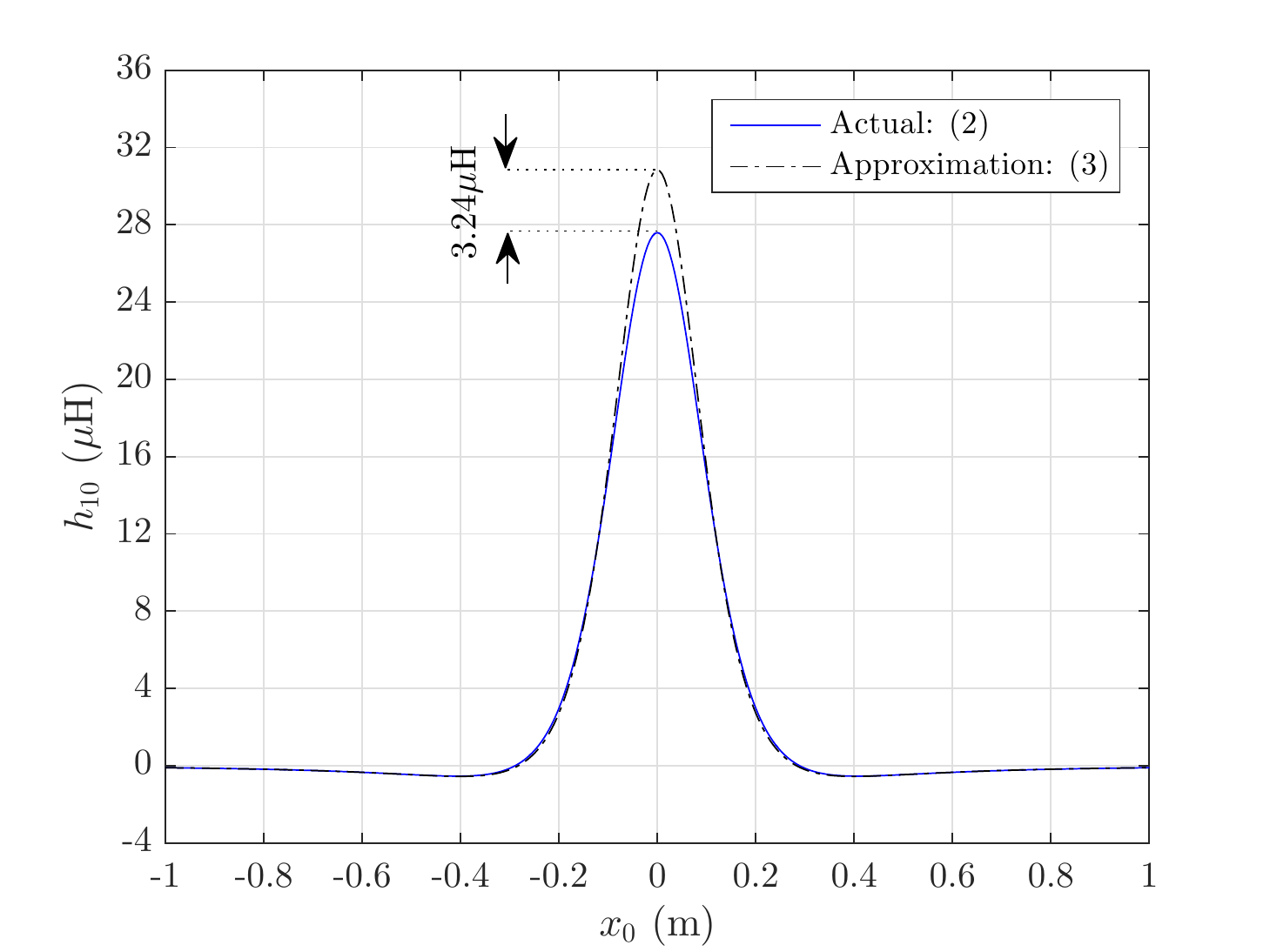}}}  
		\subfigure[$z_0=0.4$m]
		{\scalebox{0.52}{\includegraphics{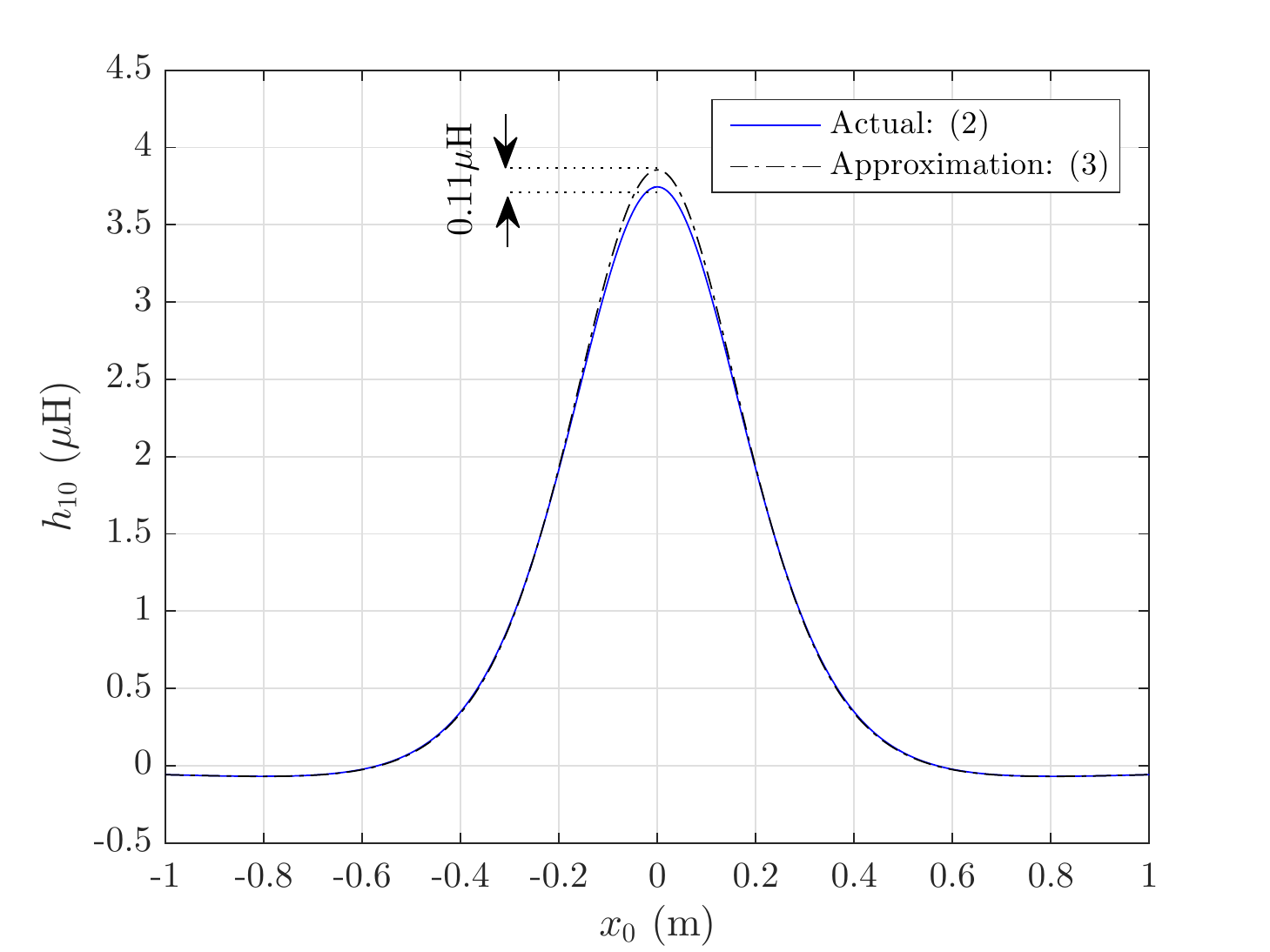}}}
	\end{center} \vspace{-4mm}
	\caption{Actual versus  approximated mutual inductance.}
	\label{fig:Mutual_Inductance_veification} \vspace{-3mm}
\end{figure}
The similar result can be obtained for the mutual inductance  between other transmitters and the receiver, while the peak value of the mutual inductance shifts over the $x$-axis accordingly, i.e., it moves from $x_0=x_1=0$ to $x_0=x_n$ when transmitter $n$ is considered instead of transmitter $1$.  
In this paper, we  use the approximation in (\ref{eq:hn0-simplified}) to formulate the node placement optimization problems in Sections \ref{Sec:Node Placement Optimization} and \ref{Sec:Node Placement Optimization 2D}, while the actual values given by  (\ref{eq:h_n0}) are used for all simulations to achieve the best accuracy. The approximation in (\ref{eq:hn0-simplified}) is acceptable  in our case, since the design objective is to maximize the minimum transferable power to the receiver load, while such minimum occurs when the receiver is sufficiently away from all transmitters.   
In this case, from Fig. \ref{fig:Mutual_Inductance_veification} it is observed that the approximation is indeed much tighter when $x_0$ deviates from zero.  

By applying Kirchhoff's circuit laws to the electric circuits of the transmitters and  receiver in our considered MRC-WPT system shown in Fig. \ref{fig:Electric-Circuit}, we obtain \vspace{-0.1mm} 
\begin{align} \label{eq:in}
&r_\text{tx} i_n - j w h_{n0} i_0 +j w \hspace{-2mm}\sum_{k=1,k\neq n}^{N} \hspace{-2mm}h_{nk} i_k=v_n,~ n=1,\ldots,N,\\
&r_\text{rx} i_0 - j w \sum_{n=1}^{N} h_{n0} i_n=0.\hspace{-2mm} \label{eq:i0_in_relation}
\end{align}

Let $p_n$ and $p_0$ denote the power drawn from the energy source at transmitter $n$ and that delivered to the  load at the single receiver. 
In practice, we have $p_n=\real\left\{v_n i_n^*\right\}$  and $p_0=r_{\text{rx,l}} \left|i_0\right|^2$, where $(\cdot)^*$ denotes the conjugate of a complex number.   
With the results in (\ref{eq:in}) and (\ref{eq:i0_in_relation}), we thus obtain 
\begin{align} 
&p_n= \bigg(r_\text{tx} + \dfrac{w^2}{r_\text{rx}} h_{n0}^2\bigg)|i_n|^2+ \dfrac{w^2}{r_\text{rx}}\hspace{-1mm}\sum_{k=1,k\neq n}^{N} \hspace{-2.5mm} h_{n0}h_{k0}\bigg(\overline{i}_n\overline{i}_k+\hat{i}_n\hat{i}_k\bigg) \nonumber\\
&+w \hspace{-1mm}\sum_{k=1,k\neq n}^{N} \hspace{-2.5mm} h_{nk} \bigg(\hat{i}_n\overline{i}_k-\overline{i}_n\hat{i}_k\bigg), \label{eq:p_n}\\ 
& p_0= \dfrac{w^2 r_{\text{rx,l}} }{r_\text{rx}^2}  \bigg( \bigg(\sum_{n=1}^{N} h_{n0} \overline{i}_n\bigg)^2+\bigg(\sum_{n=1}^{N} h_{n0} \hat{i}_n\bigg)^2\bigg).\label{eq:p_0}
\end{align}  
Accordingly, the sum-power drawn from all transmitters' sources can be derived as 
\begin{align} \label{eq:p_sum}
&p_{\text{sum}}=\sum_{n=1}^{N}p_n=\nonumber\\
&r_\text{tx}\sum_{n=1}^{N}|i_n|^2+  \dfrac{w^2}{r_\text{rx}} \bigg( \bigg(\sum_{n=1}^{N} h_{n0}\overline{i}_n\bigg)^2 \hspace{-0.5mm}+\bigg(\sum_{n=1}^{N} h_{n0}\hat{i}_n\bigg)^2 \bigg).
\end{align}
\begin{figure} [t!]
	\centering
	\includegraphics[scale=0.545]{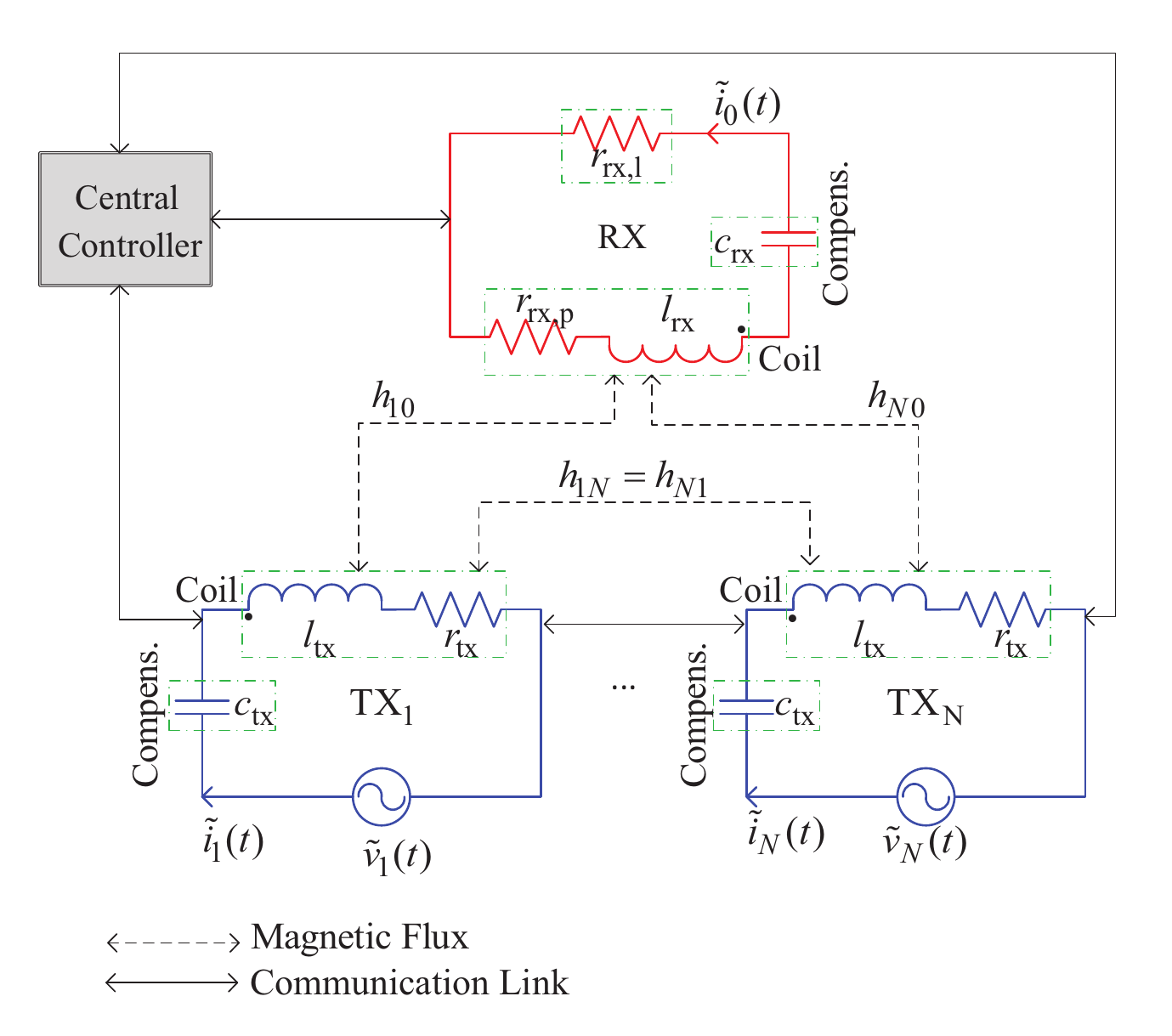}
	\vspace{-3mm}
	\caption{The electric circuit model of a MISO MRC-WPT system.} 
	\label{fig:Electric-Circuit}  \vspace{-2mm}
\end{figure}

From  (\ref{eq:p_n}) and (\ref{eq:p_sum}), it follows that with fixed $i_n$'s, the  power consumption of each individual transmitter depends on all the mutual inductances between the transmitters and the receiver, $h_{n0}$'s, as well as those between any pair of transmitters, $h_{kn}$'s, while their sum-power consumed depends on $h_{n0}$'s only.  
From (\ref{eq:p_0}) and (\ref{eq:p_sum}), it also follows that the real-part currents $\overline{i}_n$'s and the imaginary-part  currents $\hat{i}_n$'s contribute in the same way to $p_0$ or  $p_{\text{sum}}$. 
Therefore, in this paper,  we set $\hat{i}_n=0$, $n=1,\ldots,N$, and focus on designing  $\overline{i}_n$'s.\footnote{Note that in this paper, we only consider the sum-power constraint for all transmitters. In the case that the peak power constraint for each individual transmitter is considered, both $\bar{i}_n$'s and $\hat{i}_n$'s should be optimized jointly \cite{Gang}.}   
Moreover, since each $h_{n0}$ is a function of both $x_0$ and $x_n$ with given $z_0$ (see, e.g., (\ref{eq:hn0-simplified})), we re-express $p_0$ and $p_{\text{sum}}$ in (\ref{eq:p_0}) and (\ref{eq:p_sum}) as functions of $x_0$, $x_n$'s, and $\overline{i}_n$'s as  
\begin{align} 
&\hspace{-2mm}p_0(x_0,\{x_n\},\{\overline{i}_n\})=\dfrac{w^2 r_{\text{rx,l}} }{r_\text{rx}^2}   \bigg(\sum_{n=1}^{N} h_{n0} \overline{i}_n\bigg)^2,\label{eq:Fun_p_0}\\
&\hspace{-2mm}p_{\text{sum}}(x_0,\{x_n\},\{\overline{i}_n\})= r_\text{tx}\sum_{n=1}^{N} \overline{i}_n^2+ \dfrac{w^2}{r_\text{rx}} \bigg(\sum_{n=1}^{N} h_{n0}\overline{i}_n\bigg)^2 .\label{eq:Fun_p_sum}
\end{align} 

Next,  we introduce four metrics to evaluate the performance of the  MRC-WPT system, which are the average value, the minimum value, the maximum value, and the min-max ratio of the deliverable power to the receiver load over the target line (or target region in general),  defined as  
\begin{align}
&p_{0,\text{avg}}(\{x_n\},\{\overline{i}_n\})=\int_{-d}^{d} p_0(x_0,\{x_n\},\{\overline{i}_n\}) dx_0, \label{eq:p0ave}\\
&p_{0,\min}(\{x_n\},\{\overline{i}_n\})=\min_{|x_0|\le d}~ p_0(x_0,\{x_n\},\{\overline{i}_n\}), \label{eq:p0min} \\
&p_{0,\max}(\{x_n\},\{\overline{i}_n\})=\max_{|x_0|\le d}~p_0(x_0,\{x_n\},\{\overline{i}_n\}), \label{eq:p0max}\\ 
&\xi(\{x_n\},\{\overline{i}_n\}) = \dfrac{ p_{0,\min}(\{x_n\},\{\overline{i}_n\})}{ p_{0,\max}(\{x_n\},\{\overline{i}_n\})}.  \label{eq:min-max-ratio}
\end{align} 
Note that both the transmitter currents $\overline{i}_n$'s and the transmitter locations $x_n$'s can  influence the above performance metrics; as result, we need to design them jointly to optimize each corresponding performance in general.

In practice, it is desirable to have both large $p_{0,\text{avg}}$ and $p_{0,\max}$ to maximize the WPT efficiency, and yet have acceptably high $p_{0,\min}$ and $\xi$ to achieve uniform performance over the target region. 
However, in general, there are trade-offs in achieving these objectives at the same time, e.g., maximizing $p_{0,\max}$ versus $p_{0,\min}$. 
In the rest of this  paper,  we first design the magnetic beamforming via  adjusting  $\overline{i}_n$'s by assuming fixed locations of the transmitters and receiver ($x_n$'s and $x_0$) to maximize the deliverable power subject to a given sum-power constraint for all transmitters.  
Next, with the obtained optimal magnetic beamforming solution, we optimize the transmitter locations  $x_n$'s to maximize the minimum power deliverable to the receiver over the target region, i.e.,  $p_{0,\min}$ given in (\ref{eq:p0min}), for both the cases of 1D and 2D target regions, respectively.     
%*************************
%*************************
%*************************
%*************************
\section{Optimal Magnetic Beamforming}   \label{Sec:magneticBeamforming}
In this section, we first present the results on  the magnetic beamforming optimization.  
We  then use a numerical example to demonstrate the performance advantages of optimal distributed magnetic beamforming. 
%*************************
\vspace{-2mm}
\subsection{Problem Formulation and Solution} \label{Sec:problem Formulation} 
Assume that $x_n$'s and $x_0$ are given, and hence the mutual inductance values $h_{n0}$'s are known.  
We formulate the  magnetic beamforming problem for designing  the transmitter currents $\overline{i}_n$'s to maximize the deliverable power to the receiver load,  $p_{0}$ given in (\ref{eq:Fun_p_0}), subject to a  maximum sum-power constraint at all transmitters, denoted by $p_{\max}>0$, as follows. 
\begin{align} 
\mathrm{(P1)}:  
\mathop{\mathtt{max}}_{\{\overline{i}_n\}}~&\dfrac{w^2 r_\text{rx,l}}{r_\text{rx}^2}  \bigg(\sum_{n=1}^{N} h_{n0}\overline{i}_n\bigg)^2  \label{eq:p1_obj} \\
\mathtt{s.t.} ~&  r_\text{tx}\sum_{n=1}^{N} \overline{i}_n^2+\dfrac{w^2}{r_\text{rx}}   \bigg(\sum_{n=1}^{N} h_{n0}\overline{i}_n\bigg)^2\le p_{\text{max}}. \label{eq:p1_c1}
\end{align}
(P1) is a non-convex quadratically constrained quadratic programming (QCQP) problem \cite{ConvOpt-Book}, since its objective is to maximize a convex quadratic function in (\ref{eq:p1_obj}).  
However, we obtain the optimal solution to (P1) in the following proposition. 
\begin{proposition} \label{Proposition:Optimal_Current}
The optimal solution to (P1) is given by  $\overline{i}_n^\star$, $n=1,\ldots,N$, with \vspace{-2mm}   
\begin{align} 
\overline{i}_n^\star=\dfrac{h_{n0}\sqrt{p_{\max}}}{\sqrt{  \bigg(\sum_{k=1}^{N} h_{k0}^2 \bigg) \bigg(r_\text{tx}+\dfrac{w^2}{r_\text{rx}} \sum_{k=1}^{N} h_{k0}^2 \bigg) }}.  \label{eq:u_n}
\end{align}
\begin{proof}
Please see Appendix B.
\end{proof}
\end{proposition}

\textit{Remark}: The magnetic beamforming for the case of general complex-valued mutual inductance and/or unequal transmitter parameters is also  recently investigated in \cite{Lang,Gang}. However, in Proposition \ref{proposition:QN}, the beamforming solution is derived for the special case of identical transmitters with real-valued mutual inductance, which has a simpler structure compared to that in \cite{Lang,Gang} and thus  facilitates our node placement design to be discussed later in Sections \ref{Sec:Node Placement Optimization} and \ref{Sec:Node Placement Optimization 2D}.  

From (\ref{eq:u_n}), it follows that the current allocated to each transmitter $n$ is  proportional to the mutual inductance between its coil and that of the receiver, $h_{n0}$. 
This is due to the fact that in (\ref{eq:u_n}), the denominator  is the same for all transmitters, while only the numerator changes linearly with  $h_{n0}$.  Moreover, it can be seen that when there exists an $n$ such that  $|h_{n0}|\gg |h_{k0}|$, $\forall k\neq n$, then  $\overline{i}^\star_{k}\approx 0$. 
This means that all transmit power is allocated to  transmitter $n$ which has the dominant mutual inductance magnitude  with the receiver (e.g., when the receiver is directly above transmitter $n$ and more far apart from its adjacent transmitters), i.e., the transmitter selection technique \cite{BHChoi} is optimal.   
To implement the optimal magnetic beamforming  solution in practice, each transmitter $n$ needs to estimate the mutual inductance between its coil and that of the receiver, $h_{n0}$,  in real time \cite{MRVM-Zhang-1}, and send it to a central controller via e.g. the Bluetooth communication considered in the Rezence specification \cite{Rez}. 
Given the information received from all  transmitters, the central controller computes the optimal transmitter currents $\overline{i}_n^\star$'s and sends them to the individual transmitters for implementing distributed magnetic beamforming.  
As shown in Fig. \ref{fig:Electric-Circuit}, it is more convenient to use voltage source than current source at each transmitter  in practice.  
With  $\overline{i}_n^\star$'s derived, the optimal receiver current $i_0^{\star}$ can be obtained from (\ref{eq:i0_in_relation}) as $i_0^{\star}=(j w \sum_{n=1}^{N} h_{n0} \overline{i}_n^{\star})/r_{\text{rx}}$. 
By substituting $\overline{i}_n^\star$'s and $i_0^{\star}$ into (\ref{eq:in}), one can compute the optimal source voltages $v_n^\star$'s  that generate the optimal currents $\overline{i}_n^\star$'s and $i_0^{\star}$ at the transmitters and receiver, respectively, for practical implementation. Specifically, for each transmitter $n$, we have 
\begin{align}
v_n^{\star}=r_\text{tx} \overline{i}_n^{\star} +  \dfrac{w^2 h_{n0} \sum_{k=1}^{N} h_{k0} \overline{i}_k^{\star}}{r_{\text{rx}}} +j w \hspace{-2mm}\sum_{k=1,k\neq n}^{N} \hspace{-2mm}h_{nk} \overline{i}_k^{\star}.
\end{align}
Note that the obtained  $\overline{i}_n^\star$'s, $v_n^\star$'s, and $i_0^\star$ always satisfy the Kirchhoff's circuit laws given in (\ref{eq:in}) and (\ref{eq:i0_in_relation}). 

Next, by substituting  $\overline{i}_n=\overline{i}_n^\star$, $n=1,\ldots,N$, in (\ref{eq:Fun_p_0}), the power delivered to the load with optimal magnetic beamforming is given by \vspace{-1mm}
\begin{align}
&p_0^\star(x_0,\{x_n\})=p_0(x_0,\{x_n\},\{\overline{i}_n^\star\})=\nonumber\\
&\dfrac{r_\text{rx,l}}{r_\text{rx}} \left(1-\dfrac{1}{1+ \dfrac{w^2}{r_\text{rx}r_\text{tx}}\sum_{n=1}^{N}h_{n0}^2}\right)  p_{\max}. \label{eq:p_0^*}
\end{align}
From (\ref{eq:p_0^*}), it follows that the deliverable power with optimal magnetic beamforming is a function of $h_{n0}^2$'s, hence invariant to the signs of individual  $h_{n0}$'s.  
This is expected, since magnetic beamforming  ensures that the magnetic fields generated by different transmitters are constructively added at the receiver. 
%*************************
\subsection{Numerical Example} \label{Sec:numerical Example} 
We consider an MRC-WPT system with $N=5$ identical transmitters and a single receiver that is connected a load with resistance $r_{\text{rx,l}}=100\Omega$.
The physical and electrical characteristics of coils in the transmitters and receiver are given in Tables \ref{tab:Physical-Charac} and \ref{tab:Electrical-Charac}, respectively.
\begin{table}[t!]
	\centering 
	\caption{Physical characteristics of coils}\vspace{-1mm}
	\begin{tabular}{|o|d|o|d|u|}
		\hline
		\vspace{.1mm}Coil & Radius $e_{\text{coil},\text{tx}}$/$e_{\text{coil},\text{rx}}$ (mm) & Number of turns ~~~~~~~~~  $b_{\text{tx}}$/$b_{\text{rx}}$& Wire size  $e_{\text{wire},\text{tx}}$/$e_{\text{wire},\text{rx}}$ (mm) & \hspace{-2mm}Wire resistivity $\sigma_\text{tx}$/$\sigma_\text{rx}$ ~~~~~~~~~($\Omega$/m) \bigstrut\\ \hline
		Transmitter &$50$   &  $400$     & 0.1 &$1.68\times 10^{-8}$ \bigstrut\\ \hline
		Receiver     & $25$&  $200$    & 0.1 &$1.68\times 10^{-8}$   \bigstrut\\ \hline
	\end{tabular}
	\label{tab:Physical-Charac} 	
\end{table}
\begin{table}[t!]
	\centering
	
	\caption{Electrical characteristics  of coils} \vspace{-1mm}
	\begin{tabular}{|o|S|S|S|S|}
		\hline
		\vspace{.1mm}Coil &  Internal resistance $r_\text{tx}$/$r_{\text{rx,p}}$ ($\Omega$) & Self-inductance ~~~~~~~$l_\text{tx}$/$l_\text{rx}$ ~~~~~~~~~~(mH) &  Series compensator $c_\text{tx}$/$c_\text{rx}$ (fF)  \bigstrut\\ \hline
		Transmitter &  $67.20$ &  $63.27$ & $8.71$ \bigstrut\\ \hline
		Receiver     &  $16.80$ &  $7.04$ & $78.29$ \bigstrut\\ \hline
	\end{tabular} \vspace{-2mm}
	\label{tab:Electrical-Charac}
\end{table} 
The material of wire used for manufacturing coils is assumed to be copper.  We set $z_0=0.2$m, $d=1$m (i.e., the line length is $2$m in total), $w=42.6\times 10^6$rad/sec  \cite{Tseng-Novak}, and  $p_{\max}=30$W.  
In this example, we assume that transmitters are uniformly located over $|x|\le 1$m, with $x_{1}=-x_{3}=0.5$m, $x_{2}=-x_{4}=1$m, and $x_{5}=0$. 
For performance comparison, we  consider the uncoordinated WPT with equal current allocation over all transmitters, as well as  the transmitter selection technique which only selects the transmitter with the largest mutual inductance (squared) value with the receiver for WPT with the full transmit power, $p_{\max}$. 

Fig. \ref{fig:p0_3tx_x0} compares the deliverable load power $p_0$ given in (\ref{eq:Fun_p_0}) versus the receiver location $x_0$ by three  schemes: equal (transmitter) current with uniform  (transmitter) location  (ECUL), optimal (transmitter) current with  uniform (transmitter) location (OCUL), and   transmitter selection with  uniform (transmitter) location (TSUL). 
It is observed that  ECUL in general delivers higher power to the load than OCUL and TSUL, and also achieves a larger minimum power over the receiver location $x_0$.   
It is also observed that the three schemes all  tend to deliver more power to the load when the receiver is in close proximity of one of the  transmitters at $x_0=0$, $x_0\pm0.5$m, and $x_0=\pm1$m.
Furthermore, it is observed that TSUL performs quite close to OCUL except in the middle areas between any two adjacent  transmitters, where the minimum deliverable power  occurs.  
This observation is expected since when the receiver is in the middle of any two adjacent transmitters, optimal magnetic beamforming  with both transmitters delivering power to the receiver load achieves a more pronounced combining gain as compared to  the transmitter selection with only one of the two transmitters  selected for WPT. 
\begin{figure} [t!]
	\centering
	\includegraphics[scale=0.52]{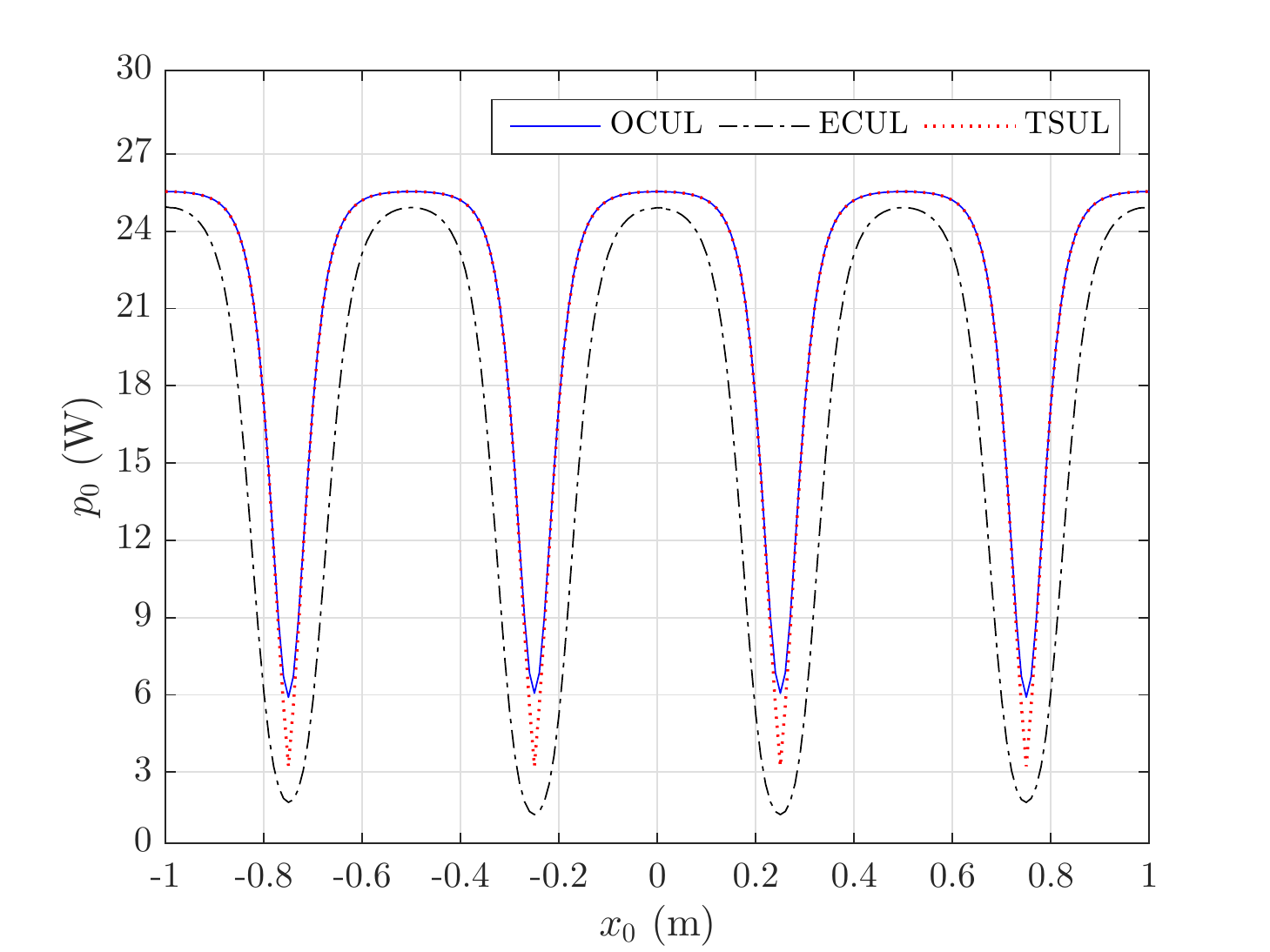} \vspace{-2mm}
	\caption{The load power profile by  distributed WPT.} 
	\label{fig:p0_3tx_x0} \vspace{-2mm}
\end{figure}
\begin{figure} [t!]
	\centering
	\includegraphics[scale=0.52]{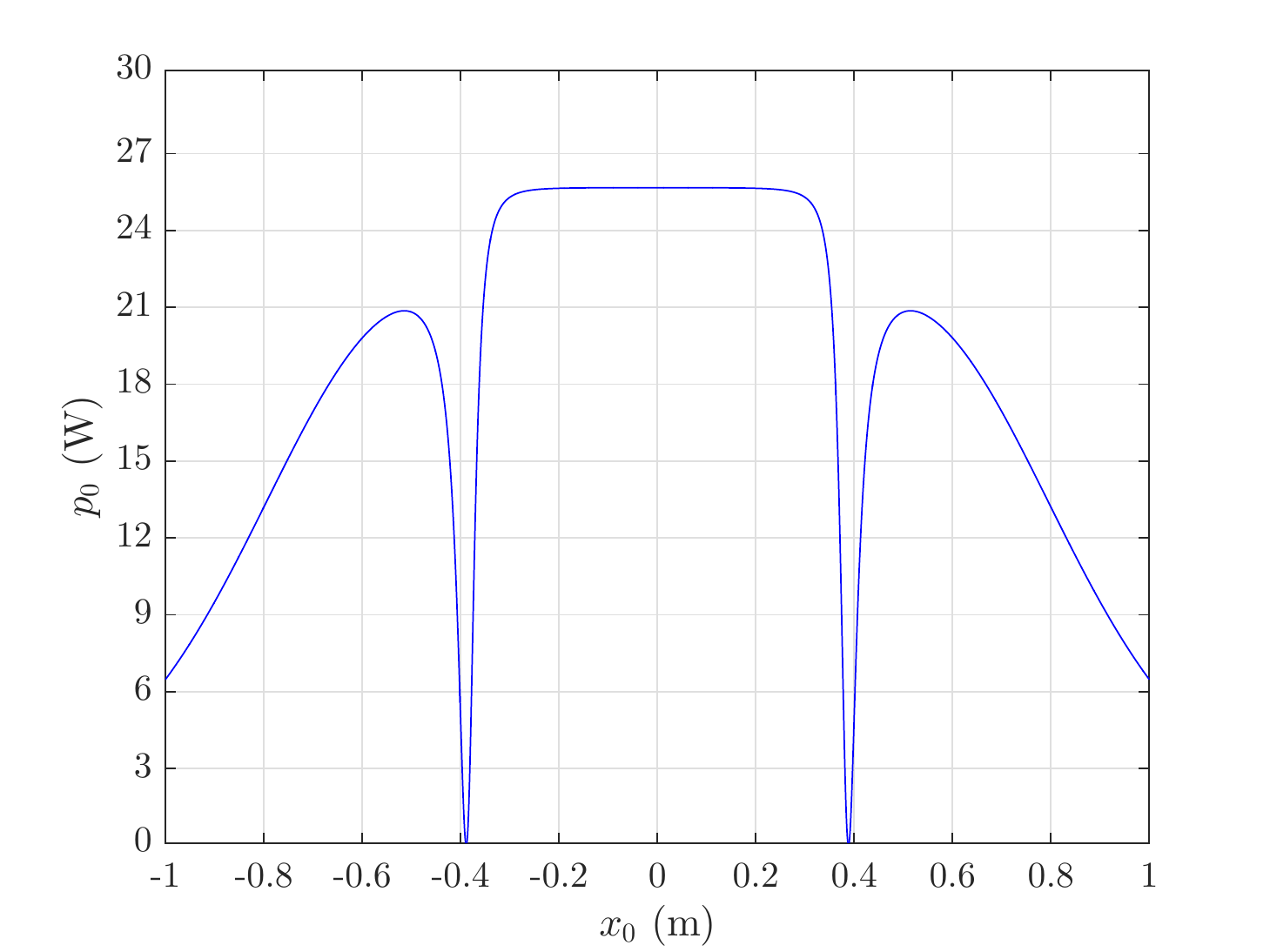} \vspace{-2mm}
	\caption{The load power profile by centralized WPT.} 
	\label{fig:p0_1tx_x0} 
\end{figure}
\begin{table}[t!]  
	\centering
	\vspace{1mm}
	\caption{Performance comparison between distrusted versus centralized WPT.} 	\vspace{-1mm}
	\label{tab:quantitiveComparision}
	\begin{tabular}{|c|c|R|R|R|R|}
		\hline
		\multicolumn{2}{|c|}{Schemes} & $p_{0,\text{avg}}$ (W) & $p_{0,\min}$ (W) & $p_{0,\max}$ (W) & $\xi$ ~~~~~~~~~~ (\%) \\ \hline
		\multirow{3}{*}{Distributed} & ~~~OCUL~~~  & $21.54$ & $5.91$ & $25.54$ & $23.14$ \\ \cline{2-6} 
		& ECUL& $16.79$ & $1.35$ & $24.94$ & $5.41$ \\ \cline{2-6} 
		& TSUL& $21.41$ & $3.23$ & $25.54$ & $12.65$ \\ \hline
		\multicolumn{2}{|c|}{Centralized} & $18.47$ & $0$ & $25.67$ & $0$ \\ \hline
	\end{tabular} \vspace{-4mm}
\end{table} 

Next, we show the performance of centralized WPT, where a single transmitter is located at $x_1=0$ which sends wireless power to a receiver moving along the target line.
For this centralized  transmitter case, we set $b_{\text{tx}}=400$  turns and $e_{\text{coil},\text{tx}}=250$mm, where the radius of its coil is $N=5$ times larger than that of each transmitter in the case of distributed WPT (i.e., $50$mm) for fair comparison. 
Fig. \ref{fig:p0_1tx_x0} plots $p_0$ for centralized WPT  versus $x_0$.   
It is observed that the deliverable power to the load is zero at $x_0=\pm 0.389$m, while its  global and local maximums occur at $x_0=0$ and $x_0=\pm0.514$m, respectively. 
Note that from (\ref{eq:h_n0}), it follows $h_{10}=0$ at $x_0=\pm 0.389$m; as a result, by setting $h_{10}=0$ in (\ref{eq:p_0}), we have $p_0=0$, regardless of the transmit current.

The details of performance  comparison between  distributed versus centralized WPT in terms of the four metrics introduced in Section \ref{Sec:System-Setup} (see (\ref{eq:p0ave})--(\ref{eq:min-max-ratio})) are  given in Table \ref{tab:quantitiveComparision}. 
It is observed that distributed WPT with OCUL  and TSUL achieves similar $p_{0,\max}$ and slightly better $p_{0, \text{avg}}$ compared to  centralized WPT. However, in terms of $p_{0,\min}$ and the min-max load power ratio $\xi$, distributed WPT achieves significant improvement over centralized WPT. 
Although distributed WPT with OCUL achieves the highest $\xi$ of $23.14\%$, it is still far from the ideal uniform power profile with      $\xi=100\%$.
To further improve this performance,  in the next section, we will formulate the node placement problem to design the transmitter locations to maximize the minimum deliverable power  to the receiver load over the target line  jointly with the optimal magnetic beamforming.
It is worth pointing out that the transmitter locations can be optimized with magnetic beamforming  to improve other performance metrics   such as maximizing the average load power, maximizing the min-max ratio of the load power, etc., which will lead to  different optimal transmitter locations in general. 
We leave other possible node placement problem formulations to our  future work.
%*************************
%*************************
%*************************
%*************************
\vspace{-1mm}
\section{Node Placement Optimization in 1D}
\label{Sec:Node Placement Optimization}
In this section, we first present the node placement optimization problem for the 1D target line, and then propose an iterative algorithm to solve it.
%***********************
\vspace{-1mm}
\subsection{Problem Formulation} 
Let $\tau$ denote the minimum deliverable  power to the load over the  target line (see Fig. \ref{fig:MISO-Linear}). 
The node placement problem is thus formulated as  \vspace{-2mm}
\begin{align}  
\mathrm{(P2)}:  \mathop{\mathtt{max}}_{\tau,~\{x_n\}}~&\tau \label{eq:p2-objective}
\\
\mathtt{s.t.} ~&p_0^\star(x_0,\{x_n\}) \ge \tau,~ |x_0|\le d, \label{eq:p2_c1} \\
~&|x_n|\le d,~n=1,\ldots,N, 
\end{align} 
with  $p_0^\star$  given in (\ref{eq:p_0^*}). First, it can be easily shown by contradiction  that the optimal solution $x_n$'s to (P2) must be  symmetric over $x=0$, as shown in Fig. \ref{fig:MISO-Linear}. 
With symmetric transmitter locations, then it  follows that the load power distribution over the target line is also symmetric over $x=0$; as a result, the  constraint (\ref{eq:p2_c1}) only needs to be considered  over $0 \le x_0 \le d$.   
With the above observations, we simplify (P2) for the cases of even and odd $N$, respectively, as follows. 
When $N$  is even, we have \vspace{-2mm}
\begin{align} 
&\mathrm{(P2-EvenN)}: \mathop{\mathtt{max}}_{\tau,~\{d_n\}}~\tau \\
\mathtt{s.t.}~&\sum_{n=1}^{M}f_{z_0}(d_n,x_0) \ge g(\tau),~ 0 \le x_0 \le d, \label{eq:const1:P2-E} \\
~&0\le d_n \le d,~n=1,\ldots,M, \label{eq:const3:P2-E} 
\end{align} 
with 
\begin{align}
&f_{z_0}(d_n,x_0)= \nonumber\\ &\dfrac{\bigg(2z_0^2-(d_n-x_0)^2\bigg)^2}{\bigg(z_0^2+(d_n-x_0)^2\bigg)^5} 
+\dfrac{\bigg(2z_0^2-(d_n+x_0)^2\bigg)^2}{\bigg(z_0^2+(d_n+x_0)^2\bigg)^5},
\end{align}
and 
\begin{align} \label{eq:gt}
g(\tau)=\left\{
\begin{array}{cl}
\dfrac{r_{\text{rx}}^2 r_{\text{tx}}\tau}{w^{2} \beta (r_{\text{rx,l}}p_{\max}-r_{\text{rx}}\tau)} & \text{if } \tau < \dfrac{r_{\text{rx,l}}}{r_{\text{rx}}} p_{\max}\\
\infty & \text{otherwise}. 
\end{array} \right. 
\end{align}
Note that since it  can be easily verified that the constraint in (\ref{eq:p2_c1}) is infeasible regardless of  $x_0$ when $\tau\ge (r_{\text{rx,l}} p_{\max})/r_{\text{rx}}$, we define $g(\tau)=\infty$ for $\tau\ge (r_{\text{rx,l}} p_{\max})/r_{\text{rx}}$ in (\ref{eq:const1:P2-E}) for convenience.  
On the other hand, when $N$  is odd, we have
\begin{align} 
&\mathrm{(P2-OddN)}: \mathop{\mathtt{max}}_{\tau,~\{d_n\}}~\tau \\
\mathtt{s.t.}~&\dfrac{f_{z_0}(0,x_0)}{2}+\sum_{n=1}^{M}f_{z_0}(d_n,x_0)\ge g(\tau),~  0 \le x_0 \le d, \hspace{-1mm}  \label{eq:const1:P2-O}\\
~&0\le d_n \le d,~n=1,\ldots,M.
\end{align} 

(P2$-$EvenN) and (P2$-$OddN) are both non-convex optimization problems  due to the constraints in (\ref{eq:const1:P2-E})  and  (\ref{eq:const1:P2-O}), respectively. 
Thus, it is difficult to solve them optimally. 
In the following, we  propose an iterative algorithm to obtain approximate solutions for them. 
%*************************
%*************************
%*************************
%*************************
\subsection{Proposed Iterative Algorithm} \label{Sec:Optimal Solution}  
In this subsection, we focus on the problem (P2$-$EvenN) for the even $N$ case, while the proposed algorithm can be similarly applied for (P2$-$OddN) in the odd $N$ case.  
In (P2$-$EvenN), we need to find the largest $\tau$, $0 \le \tau \le p_{\max}$,  under which the problem is feasible  over all possible transmitter (one-sided) locations $d_n$'s.     
To this end, we apply the bisection method to find the largest $\tau$ by using the fact that if (P2$-$EvenN) is not feasible for a certain $\dot{\tau}$, $0\le \dot{\tau} \le p_{\max}$, then it cannot be feasible for $\dot{\tau} < \tau \le p_{\max}$.
Similarly, if (P2$-$EvenN) is feasible for $\dot{\tau}$, then it must be feasible for $0\le \tau < \dot{\tau}$.  
The detail of our proposed algorithm is given in the following.

Initialize $\underline{\tau}=0$ and  $\overline{\tau}=p_{\max}$.   At each iteration,  we first set $\tau=(\underline{\tau}+\overline{\tau})/2$, and test the feasibility of
(P2$-$EvenN) given $\tau$ by considering the following feasibility problem. 
\begin{align*} 
\mathrm{(P2F-EvenN)}:\mathop{\mathtt{Find}}~&\{0\le d_n \le d,~n=1,\ldots,M \}\\
\mathtt{s.t.} ~&  (\ref{eq:const1:P2-E}). 
\end{align*} 
If (P2F$-$EvenN) is feasible, we save  its solution as $d_n^\star$, $n=1,\ldots,M$, and  update $\underline{\tau} = \tau$ to search for larger values of $\tau$ in the next iteration. 
Otherwise, if (P2F$-$EvenN) is infeasible, we update $\overline{\tau} =\tau$ to search for smaller values of $\tau$ in the next iteration. 
We stop the search when $\overline{\tau}-\underline{\tau} \le \epsilon$, where $\epsilon>0$ is a small constant controlling the algorithm accuracy.  
It can be easily shown that the algorithm converges after about $\log_2(p_{\max}/\epsilon)$ iterations.     
After convergence, we return $d_n^\star$'s as the solution to (P2$-$EvenN), and set  $x_n^\star=-x_{M+n}^\star=d_n^\star$, $n=1,\ldots,M$, as the solution to (P2) for the even $N$ case. 

Now, we focus on solving the feasibility problem (P2F$-$EvenN) at each iteration.  Since (P2F$-$EvenN) is non-convex, we use the following gradient based method to search for a feasible solution to this problem in an iterative manner. Initialize  $d_n=(2n-1)d/(N-1)$, $n=1,\ldots,M$. 
At each iteration $itr=1,2,\ldots$, given $d_n$'s, we check whether the constraint (\ref{eq:const1:P2-E}) holds or not. 
If the constraint holds, we return $d_n$'s as a feasible solution to (P2F$-$EvenN) and stop the search; otherwise, we update $d_n$'s as follows. 
First, we find $\dot{x}_0=\arg \min_{0\le x_0 \le d} \sum_{n=1}^{M}f_{z_0}(d_n,x_0)$, which can be numerically obtained with given $d_n$'s. 
Define       $f_{\text{min}}=\sum_{n=1}^{M}f_{z_0}(d_n,\dot{x}_0)$, which represents the minimum value of the summation term on the left hand side (LHS) of the constraint in (\ref{eq:const1:P2-E}) over $0 \le x_0 \le d$, with the given  $d_n$'s. Next, we have
\begin{align}
&\dfrac{\partial f_{\min}}{\partial d_n}=\dfrac{\partial f_{z_0}(d_n,\dot{x}_0)}{\partial d_n}=\nonumber\\
& -\dfrac{6 \left(8z_0^4+\left(d_n-\dot{x}_0\right)^4-6z_0^2\left(d_n-\dot{x}_0\right)^2\right) \left(d_n-\dot{x}_0\right)}{\left(z_0^2+\left(d_n-\dot{x}_0\right)^2\right)^6} \nonumber\\
&-\dfrac{6 \left(8z_0^4+\left(d_n+\dot{x}_0\right)^4-6z_0^2\left(d_n+\dot{x}_0\right)^2\right) \left(d_n+\dot{x}_0\right)}{\left(z_0^2+\left(d_n+\dot{x}_0\right)^2\right)^6}. \label{eq:gradient}
\end{align} 
Accordingly, we set  $d_n=\max\{0,d_n-\delta\}$ if  $\partial f_{\min}/ \partial d_n<0$, or  $d_n=\min \{d,d_n+\delta\}$ otherwise, $n=1,\ldots,M$,  with $\delta >0$ denoting a small step size.  
It can be easily verified that the above update helps increase $f_{\min}$ if $\delta$ is chosen to be sufficiently small.  
We repeat the above procedure for a maximum number of iterations, denoted by $itr_{\max} \ge 1$, after which we stop the search and return that (P2F$-$EvenN) is infeasible since the constraint (\ref{eq:const1:P2-E}) still does not hold with all  $d_n$'s derived. 
In practice, the performance of the gradient-based search for the feasible solution to (P2F$-$EvenN) depends on the initial values of $d_n$'s as the search in general converges to a local maximum of the LHS function of (\ref{eq:const1:P2-E}). 
\begin{table}[t!]
	\begin{center}
		\caption{Algorithm  for (P2$-$EvenN).} 
		\hrule\vspace{0.05cm} 
		\textbf{Algorithm $1$}
		\hrule 
		\begin{itemize}
			\item[a)] Initialize  $\epsilon>0$, $\delta>0$, $itr_{\max}\ge 1$, $rpt_{\max}\ge 1$, $\underline{\tau}=0$, and $\overline{\tau}=p_{\max}$.
			\item[b)] {\bf While $\overline{\tau}-\underline{\tau}> \epsilon$ do}:
			\begin{itemize}
				\item[1)] Set $\tau=(\underline{\tau}-\overline{\tau})/2$.  
				\item[2)] Set $Flag=0$, $itr=1$, $rpt=1$, and $d_n=n d/M$, $n=1,\ldots,M$.
				\begin{itemize}
					\item[$\bullet$] {\bf While} $Flag=0$, $itr\le itr_{\max}$, and $rpt\le rpt_{\max}$:
					\begin{itemize}
						\item[$\diamond$] Given $d_n$'s, check  the constraint (\ref{eq:const1:P2-E}). {\bf If} it holds, {\bf then} set $Flag=1$ and go to step 3); {\bf otherwise}, find the derivatives $\partial f_{\min}/\partial d_n$'s as in (\ref{eq:gradient}) and set   $d_n=\max\{0,d_n-\delta\}$ if $\partial f_{\min}/\partial d_n<0$, or  $d_n=\min \{d,d_n+\delta\}$ otherwise, $n=1,...,M$. 
						\item[$\diamond$] Set $itr=itr+1$.
						\item[$\diamond$] {\bf If}   $itr> itr_{\max}$ and $rpt\le rpt_{\max}$, {\bf then} reset the initial points as $d_n=\min \{d,\max \{0,(2n-1)d/(N-1)+\Delta d_n\}\}$, $n=1,\ldots,M$. Set $rpt=rpt+1$ and $itr=1$.  
					\end{itemize}
				\end{itemize}   
				\item[3)]{\bf If} $Flag=1$, {\bf then} set   $d_n^\star=d_n$, $n=1,\ldots,M$, and   $\underline{\tau} = \tau$;  {\bf otherwise} set $\overline{\tau} = \tau$.
			\end{itemize}
			\item[c)] Return $d_n^{\star}$'s as the solution to (P2$-$EvenN). 
		\end{itemize}
		\hrule \label{algorithm:1}
	\end{center} \vspace{-5mm}
\end{table} 
To improve the accuracy of the search,  if it  fails to find a feasible solution to (P2F$-$EvenN) after $itr_{\max}$ iterations, then we repeat the search with a new initial point  given by $d_n=\min \{d,\max \{0,(2n-1)d/(N-1)+\Delta d_n\}\}$, $n=1,\ldots,M$, with randomly generated $\Delta d_n$ which is uniformly distributed over $[-d/(N-1), d/(N-1)]$.  
The maximum number for the set of randomly generated initial points is limited by   $rpt_{\max} \ge 1$, and we decide  (P2F$-$EvenN) is infeasible if we fail to find a feasible solution to (P2F$-$EvenN) with all $rpt_{\max}$ sets of  initial points generated. 
In general,  a larger $rpt_{\max}$ can help improve the overall accuracy of the bisection search, but at the cost of more computational complexity.   

To summarize, the complete algorithm to solve (P2$-$EvenN) is given in Table \ref{algorithm:1}, denoted by Algorithm 1. 
%*************************
%*************************
%*************************
%*************************
\section{Simulation Results}   \label{subsection:PerformanceEvaluation}
\begin{figure} [t!]
	\centering
	\includegraphics[scale=0.52]{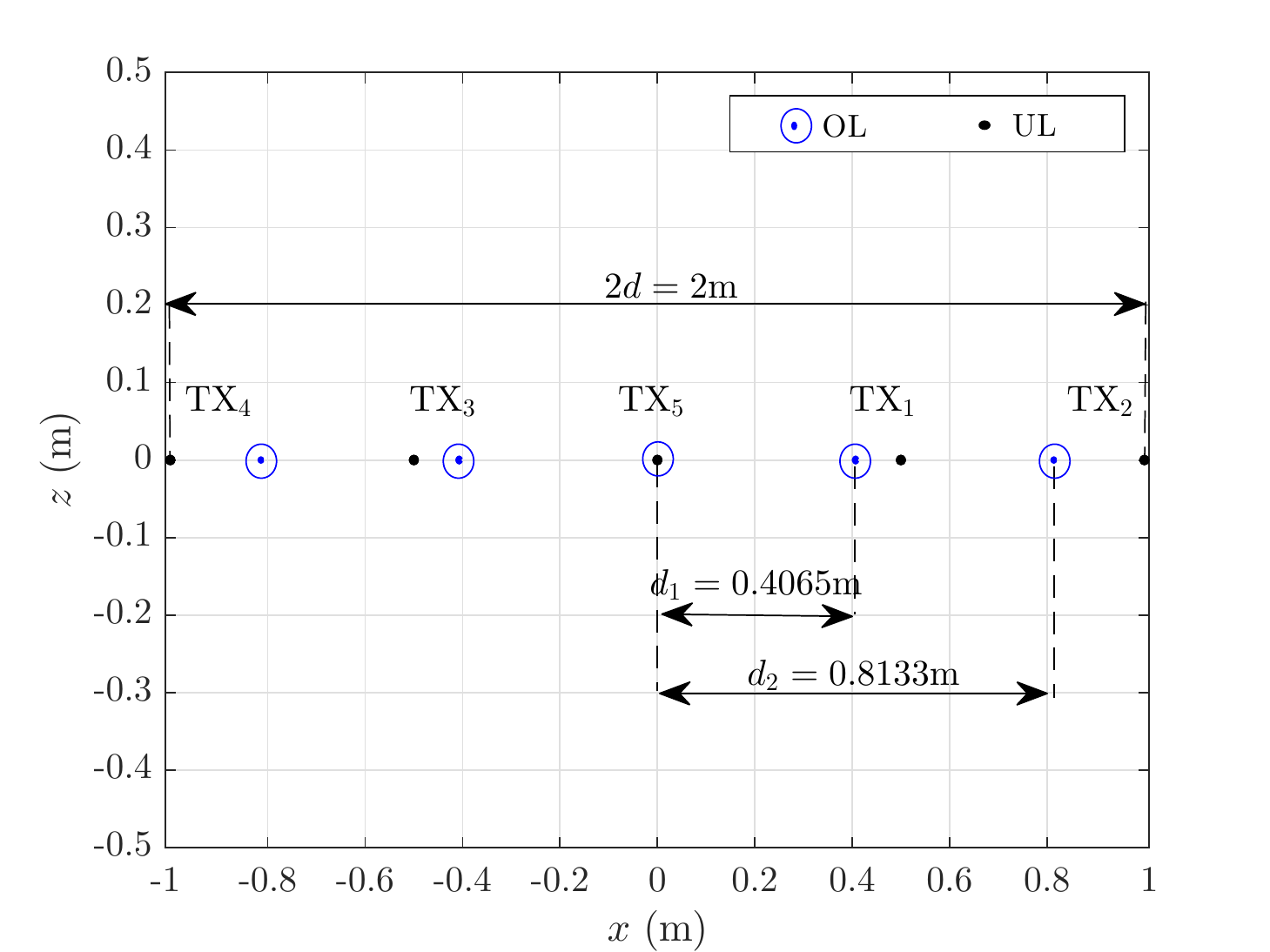} \vspace{-2mm}
	\caption{Optimized versus uniform  transmitter locations.}\label{fig:OptimalLocaion}\vspace{-2mm}
\end{figure}
\begin{figure} [t!]
	\centering
	\includegraphics[scale=0.52]{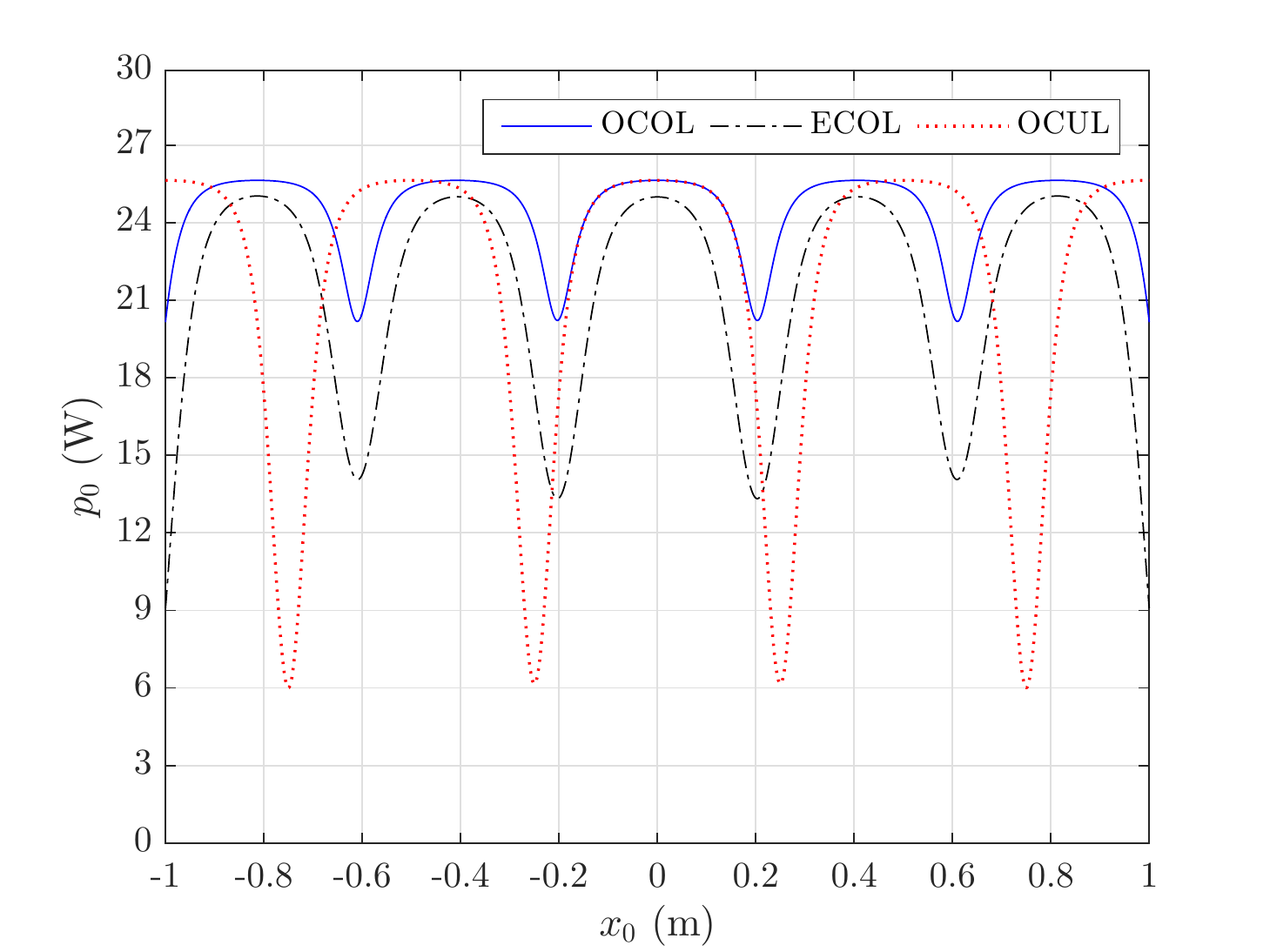} \vspace{-2mm}
	\caption{Load power profile with different transmitter locations and current   allocations.}\label{fig:NodePalacement_3tx}
\end{figure}
\begin{table}[t!] 
	\centering
	\caption{Performance  comparison between different  designs of distributed WPT.} 	\vspace{-1mm}
	\begin{tabular}{|c|o|o|o|o|o|}
		\hline
		~~~~Scheme~~~~&$p_{0,\text{avg}}$ (W)&  $p_{0,\min}$ (W)&  $p_{0,\max}$ (W)&  $\xi$\hspace{25mm} (\%)  \bigstrut\\ \hline OCOL&  $24.38$ &  $20.05$ & $25.54$& $78.50$  \bigstrut\\ \hline ECOL&  $21.31$ &  $8.93$ & $24.93$& $35.82$ \bigstrut\\ \hline OCUL&  $21.55$ &  $5.91$ & $25.54$& $23.14$ \bigstrut\\ \hline
	\end{tabular} 
	\label{tab:quantitiveComparision2}
\end{table}
In this section, we present further  simulation results to evaluate the performance of our proposed transmitter node placement algorithm, i.e., Algorithm $1$.  
We  consider the same system setup as that in Section \ref{Sec:numerical Example}, with $N=5$ identical transmitters. 
Since $N$ is odd here, we modify Algorithm $1$ for the even $N$ case to apply it for our considered system setup with $N=5$  transmitters.   We set  $\epsilon=10^{-3}$, $\delta=d/100$, $itr_{\max}=  100$, and $rpt_{\max}=100$.   

First, Fig. \ref{fig:OptimalLocaion} shows the optimized (transmitter) locations (OL), i.e.,  $x_n^\star$'s given by Algorithm $1$, versus the uniform (transmitter) locations (UL). 
It is observed that for OL, except the transmitter that is located below the center of the target line ($x=0$), the other four transmitters all move closer to  the center compared to UL.
\begin{figure} [t!]
	\centering
	\includegraphics[scale=0.52]{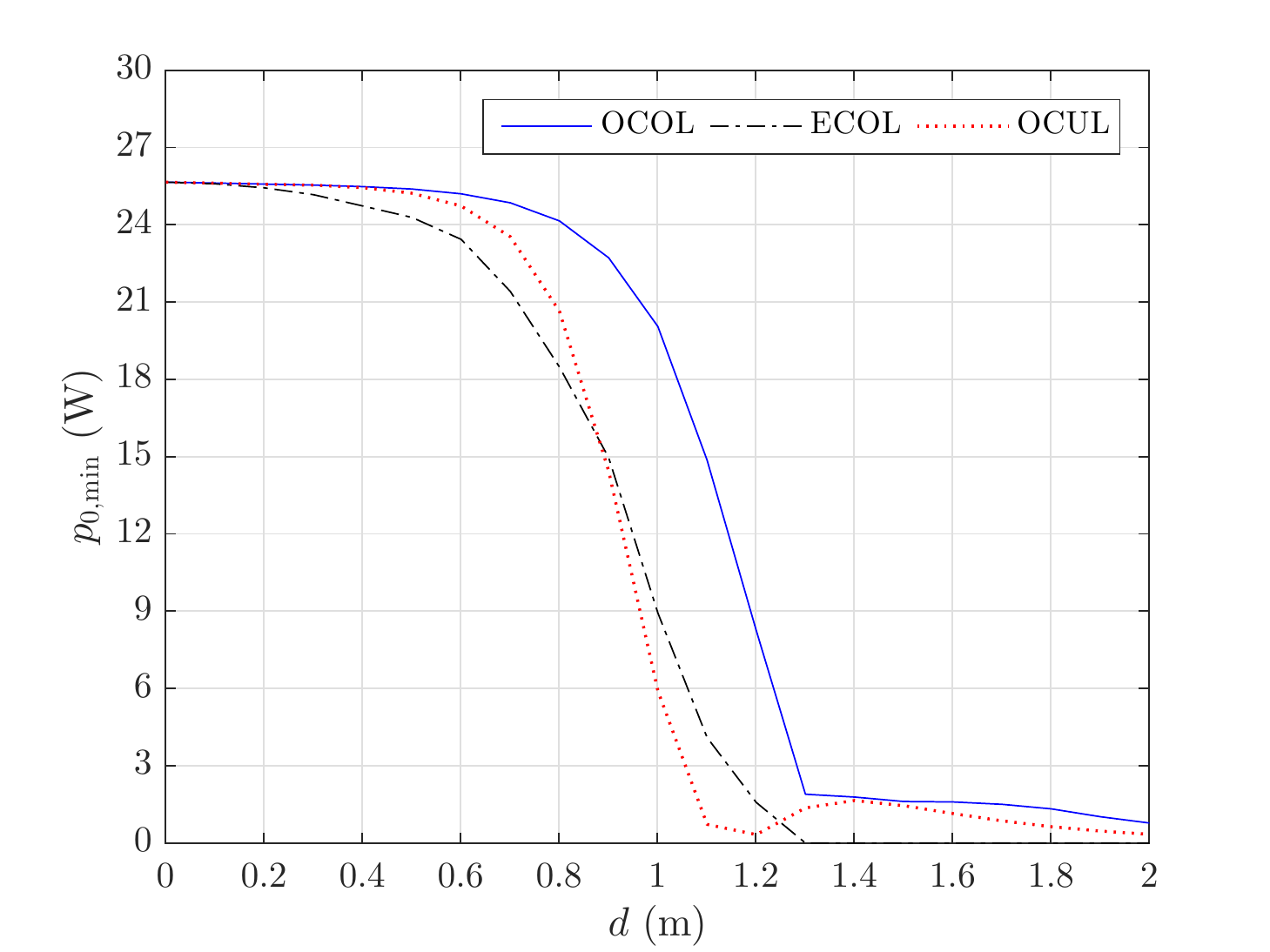}
	\vspace{-2mm}
	\caption{The minimum load power versus the target line length.} 
	\label{fig:p0min_d} \vspace{-1.5mm}
\end{figure}  
\begin{figure} [t!]
	\centering
	\includegraphics[scale=0.52]{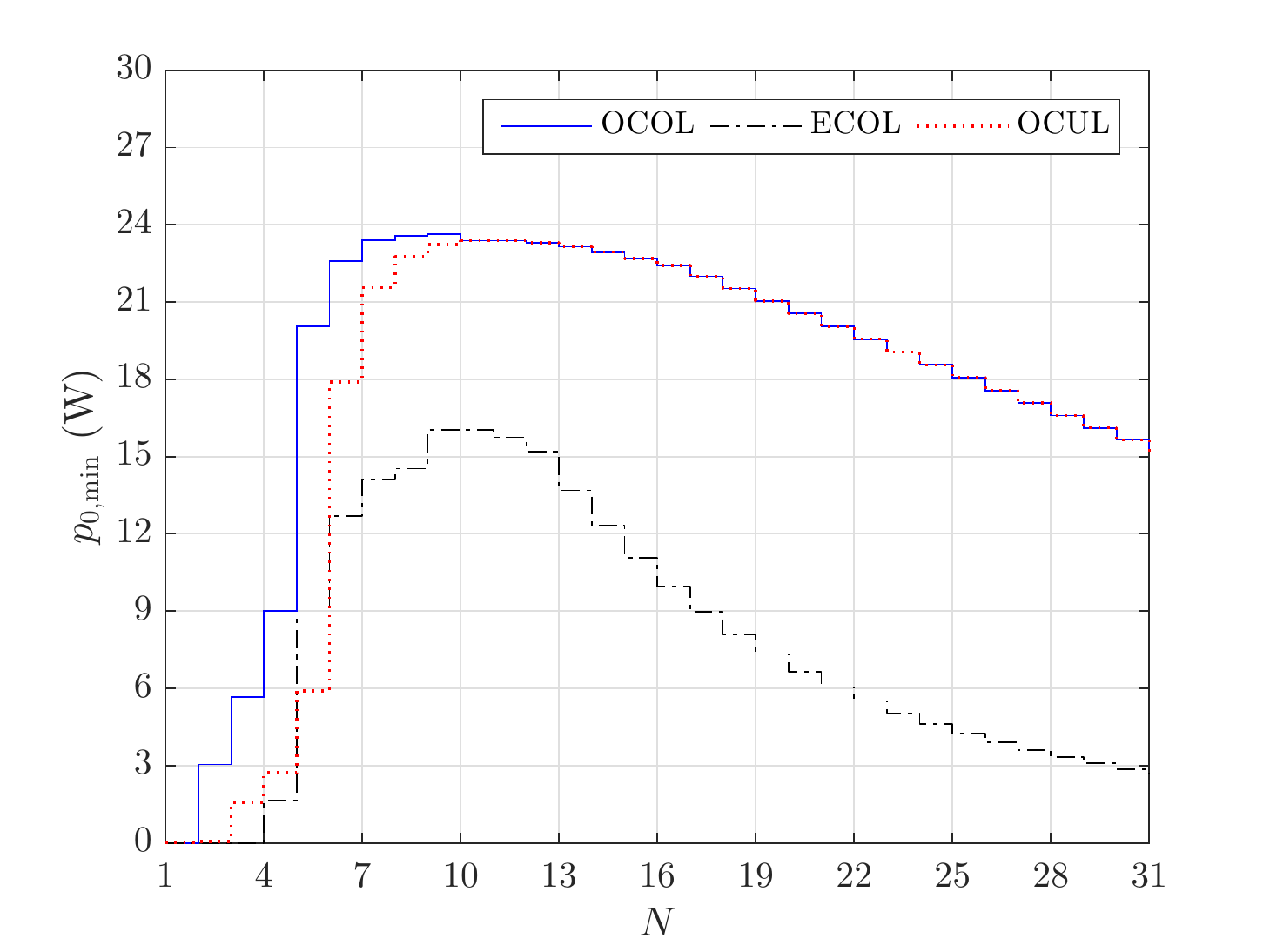}\vspace{-2mm}
	\caption{The minimum load power versus the number of transmitters with a given total coil size.} \label{fig:p0min_N} \vspace{-5mm}
\end{figure} 

Next, Fig. \ref{fig:NodePalacement_3tx}  compares  the deliverable power to the load, $p_0$ given in (\ref{eq:Fun_p_0}),  versus the receiver location  $x_0$, under three  schemes:  optimal (transmitter) current with optimized  (transmitter) location (OCOL), equal (transmitter) current with optimized (transmitter) location (ECOL), and optimal (transmitter) current with uniform (transmitter) location (OCUL).
It is observed that OCOL with both optimized transmitter locations and optimal magnetic beamforming improves the minimum deliverable power significantly over  the other two schemes with only optimized transmitter locations or optimal magnetic beamforming.   
In fact, OCOL  achieves the best performance in terms of all metrics,  where the details are given in Table \ref{tab:quantitiveComparision2}.
 
Besides, Fig. \ref{fig:p0min_d} plots the minimum deliverable power $p_{0,\min}$ given in (\ref{eq:p0min}) versus the target line length $d$, under the three schemes. 
It is observed that OCOL consistently achieves better performance  than the other two schemes, although the gain decreases when $d$ is small or large. 
This can be explained as follows. 
When $d$ is small, the mutual inductance between the receiver and different transmitters is less sensitive to  their locations, which implies that  the gain of transmitter placement optimization is small. 
In this case, from (\ref{eq:u_n}), it  follows that the transmitter currents tend to be all equal, hence the magnetic beamforming gain over the equal current allocation is also negligible.  
Similarly, when $d$ is large, the  distance between transmitters is large for both  UL and OL designs, since there are only five transmitters available to cover the target line. 
In this case, the magnetic coupling between the transmitters is small, hence they can be treated as independent transmitters. 
As shown in Fig. \ref{fig:p0_1tx_x0}, using a single transmitter for WPT cannot provide any magnetic beamforming gain. As a result, both transmitter location and current optimization do not yield notable performance gains.  
 
Last, we consider the  practical problem of finding  the optimal number of transmitters, $N$, to cover a given target line most efficiently.  
In this example, it is assumed that the total length  of coil wires for manufacturing all $N$ transmitters is fixed as $200\pi$ in meter, and thus the radius of each individual transmitter coil shrinks as $N$  increases. 
Specifically, we set the transmitter coil radius as $e_{\text{coil},\text{tx}}=250/N$ in millimeter and keep the number of the turns fixed as $b_{\text{tx}}=400$ regardless of $N$. 
The other parameters of the coils  are assumed to be the same as in Section \ref{Sec:numerical Example}. 
Fig. \ref{fig:p0min_N} plots the minimum  load power  $p_{0,\min}$ over the number of transmitters  $N$, under the aforementioned schemes of OCOL, ECOL, and OCUL. 
It is observed that for all three schemes, $p_{0,\min}$ first increases and then decreases with $N$. 
This implies that using either a small number of transmitters each with larger coil or a large number of transmitters each with smaller coil is both inefficient in maximizing the minimum deliverable power.  
Note that for the case of $N=1$, i.e., centralized WPT, $p_{0,\min}=0$, which is in accordance with the result in Fig. \ref{fig:p0_1tx_x0}. 
%*************************
%*************************
%*************************
%*************************
\section{Node Placement  Optimization in 2D}  \label{Sec:Node Placement Optimization 2D} 
\begin{figure} [t!] 
	\centering
	\includegraphics[scale=0.655]{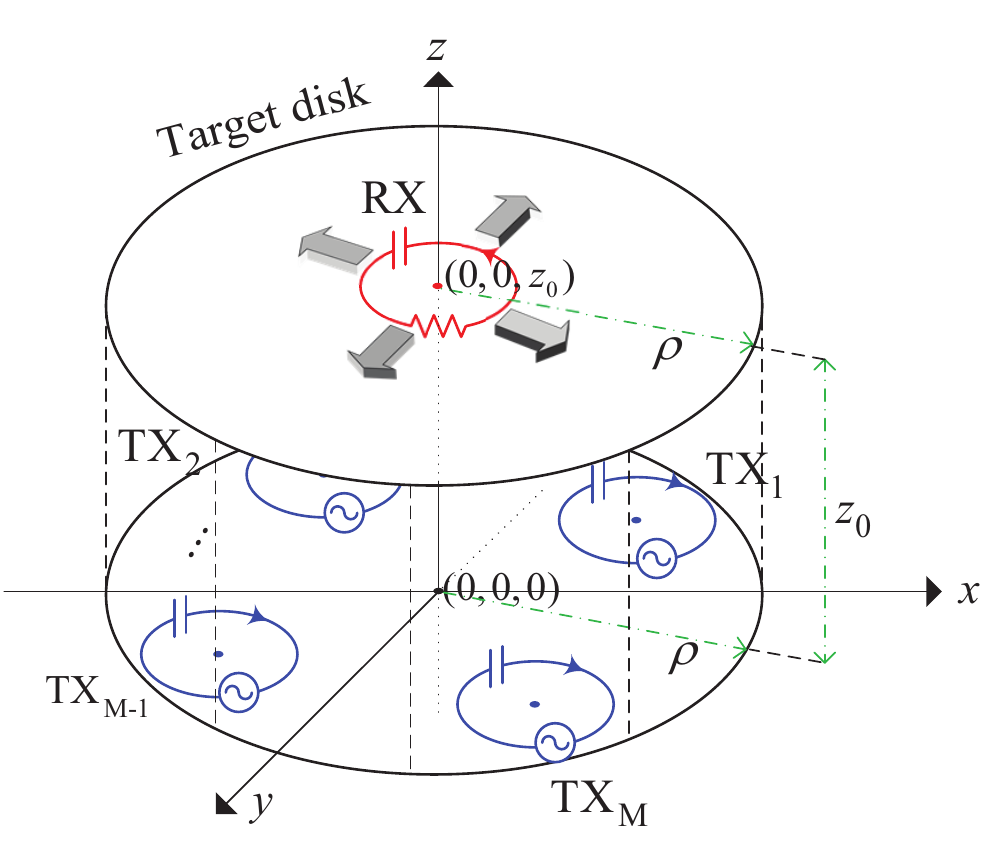} 
	\caption{The considered 2D system setup.} \label{fig:RoundTable}  \vspace{-2mm}
\end{figure} 
In this section, we extend the node placement optimization to the 2D  region case. 
As shown in Fig. \ref{fig:RoundTable}, we assume that the receiver can move horizontally within a disk of radius $\rho>0$ that  lies in the $(x,y)$ plane at a fixed height of $z=z_0$ with its  center at the origin $(x=0,~y=0,~ z=z_0)$. We denote this disk region as the 2D {\it target disk}. 
On the other hand, it is assumed that  transmitters are all horizontally placed in a disk in parallel to and below the target disk, which has the same radius of $\rho$, a fixed height of $z=0$, and  its center at the origin.   
Let $(x_n,y_n)$, with $\sqrt{x_n^2+y_n^2}\le \rho$, ($(x_0,y_0)$, with $\sqrt{x_0^2+y_0^2}\le \rho$) denote the ($x,y$)-coordinates of  transmitter $n$ (receiver).  
In this case, the mutual inductance expressions given in  (\ref{eq:h_nk}) and (\ref{eq:h_n0}) as well as  the approximation given in (\ref{eq:hn0-simplified}) can be modified by setting $d_{nk}=\sqrt{(x_n-x_k)^2+(y_n-y_k)^2}$ and $d_{n0}=\sqrt{(x_n-x_0)^2+(y_n-y_0)^2}$.
Accordingly, the transmitters' sum power and the deliverable power to the receiver load given in (\ref{eq:p_sum}) and (\ref{eq:p_0}) can be re-expressed as functions of $(x_0,y_0)$, $(x_n,y_n)$'s, and $\overline{i}_n$'s as $p_0((x_0,y_0),\{(x_n,y_n)\},\{\overline{i}_n\})$ and $p_n((x_0,y_0),\{(x_n,y_n)\},\{\overline{i}_n\})$, respectively.  
Define ${\cal R}= \{(x,y)~|~ \sqrt{x^2+y^2}\le \rho\}$, which is a convex set over $x$ and $y$. 
In general, $\cal R$ represents a generic 2D disk with a  radius of $\rho$ that lies in the $(x,y)$ plane with an arbitrary fixed height of $z=\dot{z}$ and its center  at the origin.  
In the rest of this section, when we refer to $\cal R$ as the target disk, we implicitly assume that the height is set as $\dot{z}=z_0$; otherwise, the height is $\dot{z}=0$ and $\cal R$ is used to refer to the disk region where  the transmitters are all located.  
The four performance metrics introduced for the 1D case, i.e., the average value, the minimum value, the maximum value, and the min-max ratio of the load power given in  (\ref{eq:p0ave})--(\ref{eq:min-max-ratio}), can be similarly  re-defined for the 2D case.  Specifically, each metric is a function  of $(x_n,y_n)$'s and $\bar{i}_n$'s in the 2D case.  For brevity,  the details are omitted. 
With the optimal transmitter currents given in (\ref{eq:u_n}) for magnetic beamforming, the deliverable power to the load in (\ref{eq:p_0^*}) can be then rewritten as $p_0^{\star}((x_0,y_0),\{(x_n,y_n)\})$. 

Last, note that a practical example of our considered 2D setup could be a round non-metallic  table with built-in wireless chargers mounted below its surface where the receiver can be freely placed on the table for wireless charging. In this case, $2\rho$ denotes the diameter of the table, and $z_0$ is the thickness of its surface.  
%*************************
\vspace{-2mm}
\subsection{Problem Formulation and Solution} \label{Sec:2D_ProbForm}
Similar to (P2) for the 1D case, we formulate  the node placement problem to maximize the minimum deliverable power to the receiver over the target disk $\mathcal{R}$ in 2D as
\begin{align} 
\hspace{-3mm} \mathrm{(P3)}:  \hspace{-5mm} \mathop{\mathtt{max}}_{\tau,~\{(x_n,y_n)\}}~\hspace{-3mm}& \tau \label{eq:p3-objective}
\\
\mathtt{s.t.}~&  p_0^\star((x_0,y_0),\{(x_n,y_n)\}) \ge \tau, (x_0,y_0)  \in {\cal R},\hspace{-2mm} \label{eq:p3_c1} \\
~& (x_n,y_n) \in {\cal R},~n=1,\ldots,N.  
\end{align}  
Similar to the 1D case, it can be verified that the optimal transmitter locations in (P3) must be {\it rotationally symmetric} over $\mathcal{R}$. 
In general, as shown in Fig. \ref{fig:Symmetry-Proof},  a rotationally symmetric structure for the transmitters' locations in $\cal R$ needs  to place them either at the origin and/or over one or more concentric rings, where each ring has the same center at the origin, an arbitrary radius that is less than or equal to $\rho$, consists of at least two transmitters that are equally  spaced over the ring, and has an arbitrary rotation angle with respect to the  $x$-axis.    
For $N=1$, only one rotationally symmetric structure exists by placing the single transmitter at the origin.  
For $N\ge 2$, in the following we first present a sufficient and necessary   condition to ensure that a structure consisting of $Q$ transmitter rings,  with $1\le Q \le \lfloor N/2 \rfloor$ and $ \lfloor \cdot \rfloor$ denoting the largest integer that is no greater than a given real number,  is rotationally symmetric over $\cal R$. 
Based on this condition, we then specify the total number of {\it distinct}   rotationally symmetric structures that exist for a given $N$,   denoted by $S_N\ge1$.     
For each ring $q$, $q\in \{1,\ldots,Q\}$, we denote $N_q\ge 2$ as the number of its located  transmitters, $\rho_q$, with $0 < \rho_q \le \rho$,  as its radius, and $ \phi_q$, with $0\le \phi_q\le 2\pi/N_q$, as its rotation from the $x$-axis.  
By default, we have  $\sum_{q=1}^{Q} N_q \le N$, where  $N-\sum_{q=1}^{Q}N_q$ remaining  transmitters (if any)  are all placed at $(x=0,y=0,z=0)$. 
Without loss of generality, we also set $\phi_1=0$. 
\begin{figure} [t!]
	\centering
	\includegraphics[scale=0.405]{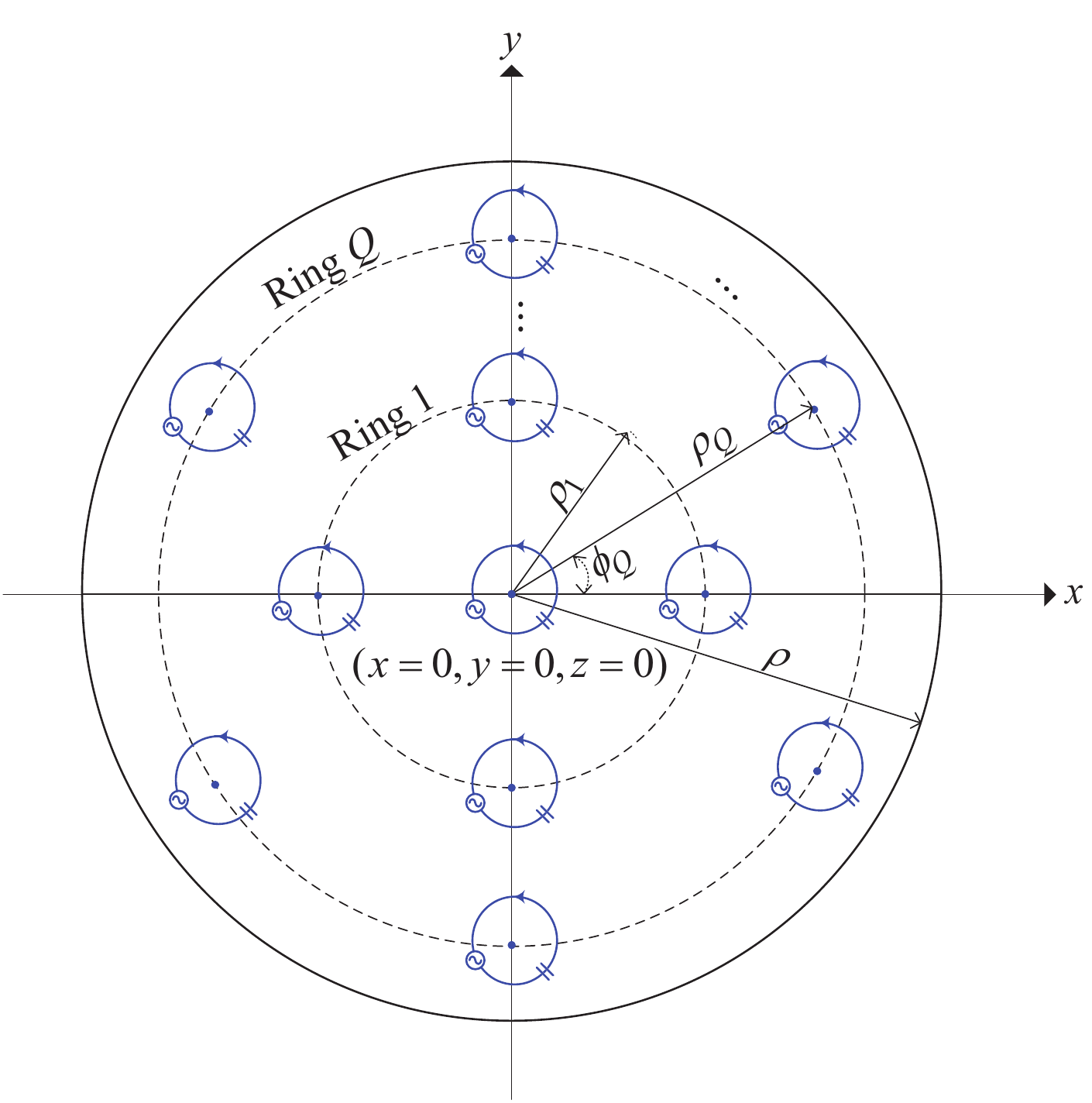}\vspace{-2mm}
	\caption{A structure consisting of $Q$ transmitter  rings.}\label{fig:Symmetry-Proof}
	\vspace{-2mm}
\end{figure}
Then, we have the following lemma. 
\begin{lemma} \label{Lemm:SymmCond}
A structure with $Q \ge 1$ transmitter rings is rotationally symmetric over $\cal R$ if and only if ({\it iff}) there exists a common divider  $u \ge 2$ such that  $N_q \mod u=0$, $\forall q=1,\ldots,Q$.
\end{lemma}
\begin{proof}
Please see Appendix C.
\end{proof}
With Lemma \ref{Lemm:SymmCond}, the following proposition thus follows. 
\begin{proposition} \label{proposition:QN}
For $N\ge 2$, we have $S_N=|\mathbb{P}_N|$, where $\mathbb{P}_N$ is the set consisting of all prime numbers less than or equal to $N$ and $|\cdot|$ denotes the cardinality of a set.
\end{proposition}
\begin{proof}
Please see Appendix D.
\end{proof}

\begin{figure*} [t!] 
	\begin{center}
		\subfigure[Structure $1$]
		{\scalebox{0.405}{\includegraphics{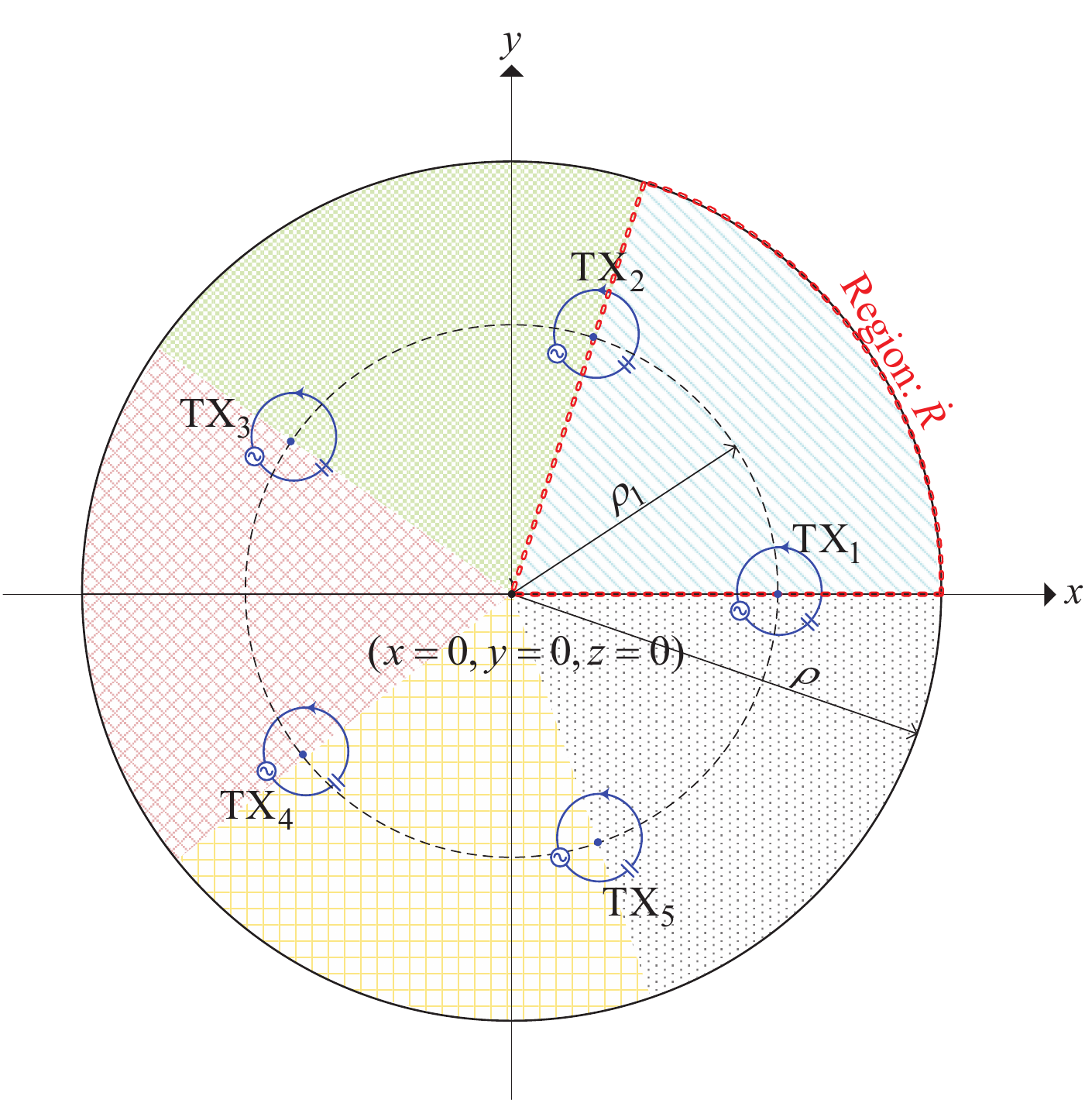}}}
		\subfigure[Structure $2$] 
		{\scalebox{0.405}{\includegraphics{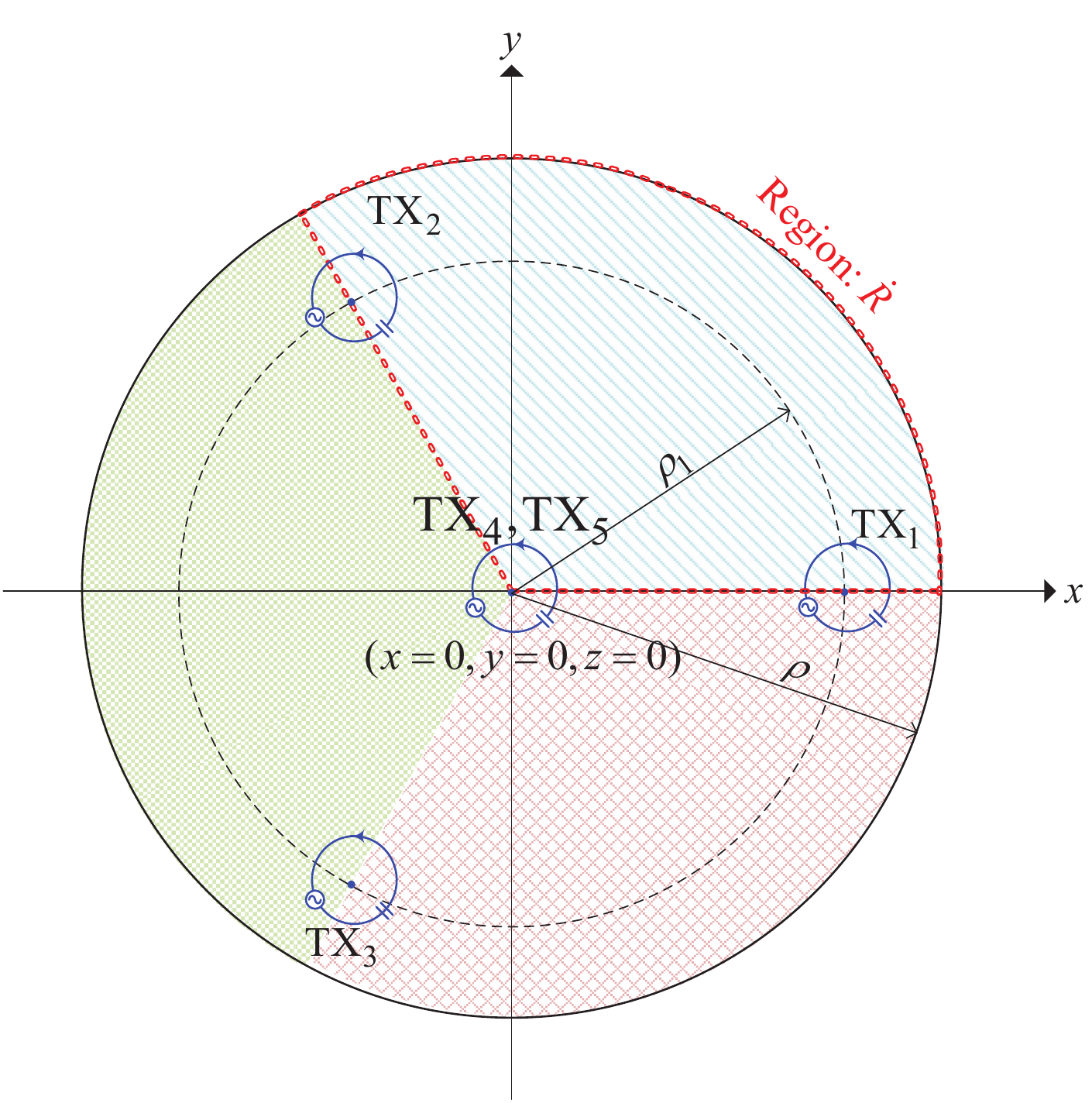}}}
		\subfigure[Structure $3$]
		{\scalebox{0.405}{\includegraphics{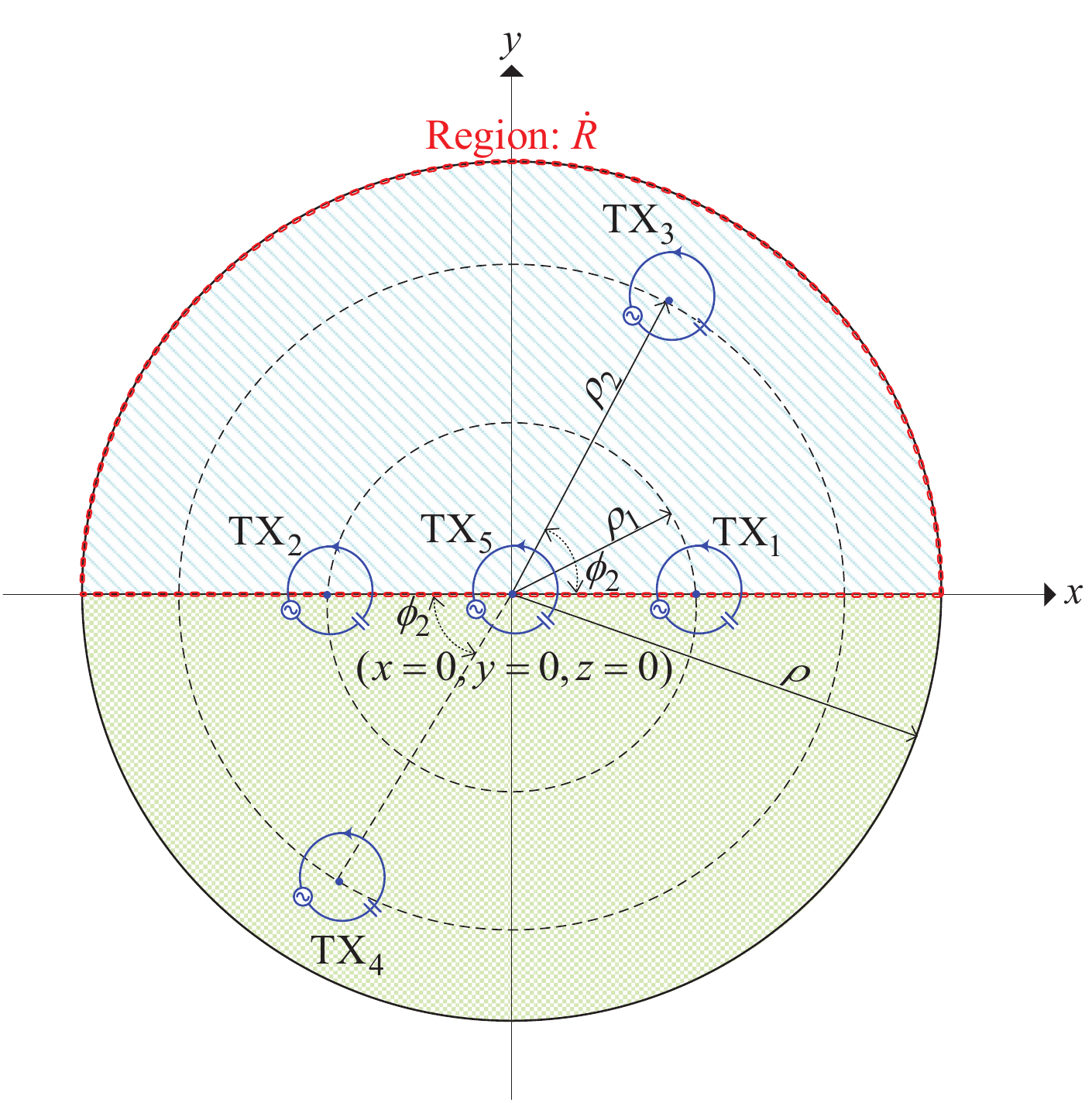}}}
	\end{center} \vspace{-1.5mm}
	\caption{Rotationally symmetric structures for a system of $N=5$   transmitters.} \vspace{-1.5mm} \label{fig:SymStructure}  
\end{figure*}   
Notice that $|\mathbb{P}_N|<N$, for $N\geq 2$, and hence the total number of rotationally symmetric structures for each given $N$ is less than $N$.  
For example, when $N=5$, from Proposition \ref{proposition:QN} it follows that $\mathbb{P}_N=\{2, 3, 5\}$, thus $S_N=3$ and  in total only three distinct  rotationally symmetric structures exist, as shown in Figs. \ref{fig:SymStructure}(a)--(c), respectively.
Moreover, different from the magnetic beamforming optimization which needs to be computed in real time according to the receiver's location, the node placement optimization can be solved offline before the transmitters are initially deployed. 
Thus, the complexity of optimizing over  $S_N$  structures to achieve the optimal transmitter placement is practically  affordable for a given $N$.  
Last, as $N$ increases, based on the so-called prime number theorem \cite{Apostol-Book}, we have asymptotically $|\mathbb{P}_N|\approx N/\ln(N)$. 
 
Next, for each rotationally symmetric  structure $s$, $s=1,\ldots,S_N$, derived from Proposition \ref{proposition:QN}, we first simplify (P3) by exploiting  the symmetry in the structure, and then solve it using a similar algorithm like Algorithm $1$ for the 1D case. 
Let  $\{(x_{n,s}^\star,y_{n,s}^\star)\}$ and $\tau_{s}^\star$ denote the optimized    transmitter locations and the resulting minimum load power for structure $s$, respectively. 
The optimal solution to (P3) is thus given by $\{(x_{n,\dot{s}}^\star,y_{n,\dot{s}}^\star)\}$, where $\dot{s}=\arg \max_{s\in \{1,\ldots,S_{N}\}} \tau_{s}^\star$.
Note that the optimal structures for different $N$ are in general not  identical. Even for a fixed $N$, the optimal structure may vary depending on the system parameters (e.g., $\rho$ as shown later in Table \ref{tab:quantitiveComparision4}).  
 
Now, we illustrate the above procedure for the case of $N=5$ transmitters, while the approach is general and can be applied to the cases with other $N$ values.
For Structure $1$ shown in Fig. \ref{fig:SymStructure}(a), (P3) is simplified as \vspace{-1mm}
\begin{align} 
&\mathrm{(P3-5TX-S1)}:  \mathop{\mathtt{max}}_{\tau,~\rho_1}~\tau \label{eq:p35TXS1-objective}
\\
\mathtt{s.t.} ~&  \sum_{n=1}^{N} \dot{f}_{z_0,\theta_n}(\rho_1,(x_0,y_0)) \ge g(\tau),~ (x_0,y_0)\in {\cal \dot{R}}, \label{eq:p35TXS1_c1} \\
~& 0\le \rho_1 \le \rho,
\end{align}
where \vspace{-1mm}  
\begin{align}
&\dot{f}_{z_0,\theta_n}(\rho_1,(x_0,y_0))=\nonumber \\ &\dfrac{\bigg(2z_0^2-\left(\rho_1 \cos(\theta_n)-x_0\right)^2-\left(\rho_1 \sin(\theta_n)-y_0\right)^2\bigg)^2}{\bigg(z_0^2+\left(\rho_1 \cos(\theta_n)-x_0\right)^2+\left(\rho_1 \sin(\theta_n)-y_0\right)^2\bigg)^5}. 
\end{align}
Moreover, we have $\theta_n=2\pi (n-1)/5$, $n=1,\ldots,5$, and ${\cal \dot{R}}=\{(x,y)~|~ \sqrt{x^2+y^2}\le \rho,~ 0\le \cos^{-1}(x/\sqrt{x^2+y^2}) \le 2\pi/5\}$, with ${\cal \dot{R}}\subset {\cal R}$ (the regions of  $\dot{\cal R}$ for Structures $2$ and $3$ are shown in Figs. \ref{fig:SymStructure}(b) and \ref{fig:SymStructure}(c), respectively).  
In (P3$-$5TX$-$S1), $\rho_1$, with $0\le \rho_1 \le \rho$, is the single decision variable (with $\tau$ as an auxiliary variable), hence Algorithm $1$ can be easily modified to solve this problem. 
Let $\rho_1^\star$ denote   the obtained solution to (P3$-$5TX$-$S1). 
Accordingly, we set $\{(x_{n,1}^\star,y_{n,1}^\star)=(\rho_1^\star \cos(\theta_n),\rho_1^\star \sin(\theta_n))\}$, $n=1,\ldots,5$, for Structure $1$. 
Similarly, we can simplify (P3) for Structure $2$ shown in Fig. \ref{fig:SymStructure}(b), for which two transmitters are placed at the origin.\footnote{In practice, the transmitter deployment shown in Fig. \ref{fig:SymStructure}(b) can be implemented by replacing transmitters $4$ and $5$ (which are co-located at the origin) by an aggregate transmitter with the same coil radius $e_{\textnormal{coil,tx}}$, but $2b_\textnormal{tx}$ turns of wire. 
Specifically, from (\ref{eq:u_n}), it follows  that the optimal currents allocated to transmitters $4$ and $5$, i.e., $i_4^{\star}$ and $i_5^{\star}$, respectively, are always identical, since $h_{40}=h_{50}$ for any receiver location.  
Hence, the two transmitters can be aggregated to a single transmitter with the aforementioned specification, without change of performance.} 
For Structure $3$ shown in Fig. \ref{fig:SymStructure}(c), we need to jointly optimize three decision variables  $\rho_1$, $\rho_2$, and $\phi_2$, with $0\le \rho_1,\rho_2 \le \rho$ and $0\le \phi_2 \le \pi$ ($\phi_1=0$ by default). The details are omitted for brevity.  
Last, note that if Structures $1$--$3$ are rotated  around the origin, their optimal solutions remain unchanged, as explained in the following. 
For example, let us rotate Structure $1$ shown in Fig. \ref{fig:SymStructure}(a) around the origin by setting $\phi_1 = \Delta \phi$, with $0<\Delta \phi< 2 \pi$ in rad.  
If $\Delta \phi=2m\pi/5$, with $m=1,\ldots,4$, then it follows that the load power distribution of the rotated Structure $1$ over the  target region $\cal R$ is the same as that for the reference  Structure $1$ with $\phi_1=0$, and hence $\rho_1^{\star}$  is the optimal solution to the rotated version of Structure $1$ as well.    
Otherwise, if $\Delta \phi \neq \{2m\pi/5~|~ m=1,\ldots, 4\}$, then the load power distribution of the rotated Structure $1$ can be simply obtained by rotating the load power distribution of the reference  Structure $1$ around the origin by $\Delta \phi$  radians.  
Obviously, the   minimum,  maximum, and average values for the deliverable load power all remain invariant. 
Hence, $\rho_1^{\star}$  is still the optimal  solution to the rotated version of Structure $1$. The similar argument is valid for Structures $2$ and $3$. As a result, in the rest of this paper we do not consider the rotated versions of Structures $1$-$3$ in our analysis/simulations. 
\begin{figure*} [t!]
	\begin{center}
		\hspace{-3mm}
		\subfigure[Distributed WPT (from left to right: Structures 1, 2, and 3)]
		{\scalebox{0.34}{\includegraphics{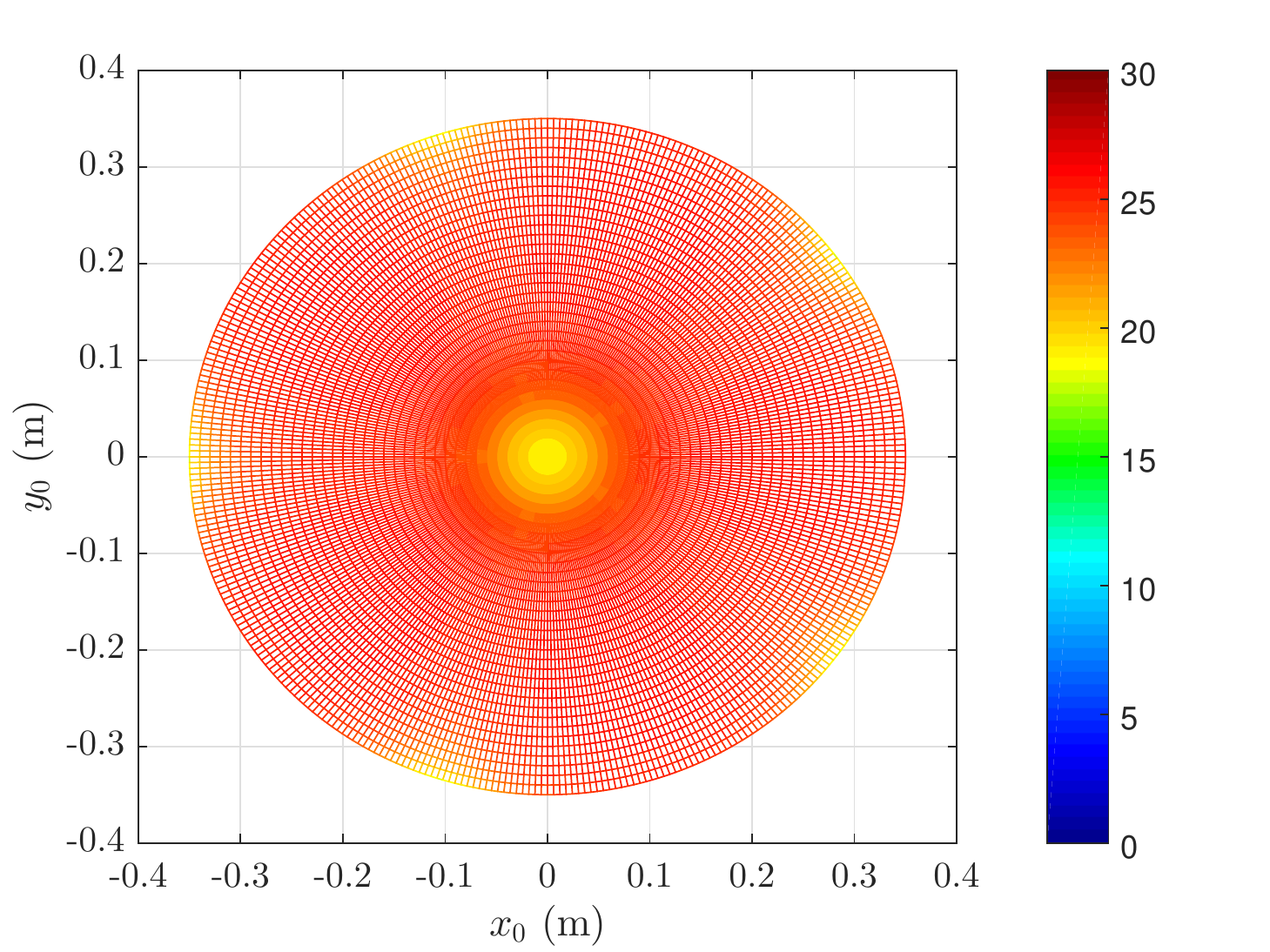}} \hspace{-13mm}
		\scalebox{0.34}{\includegraphics{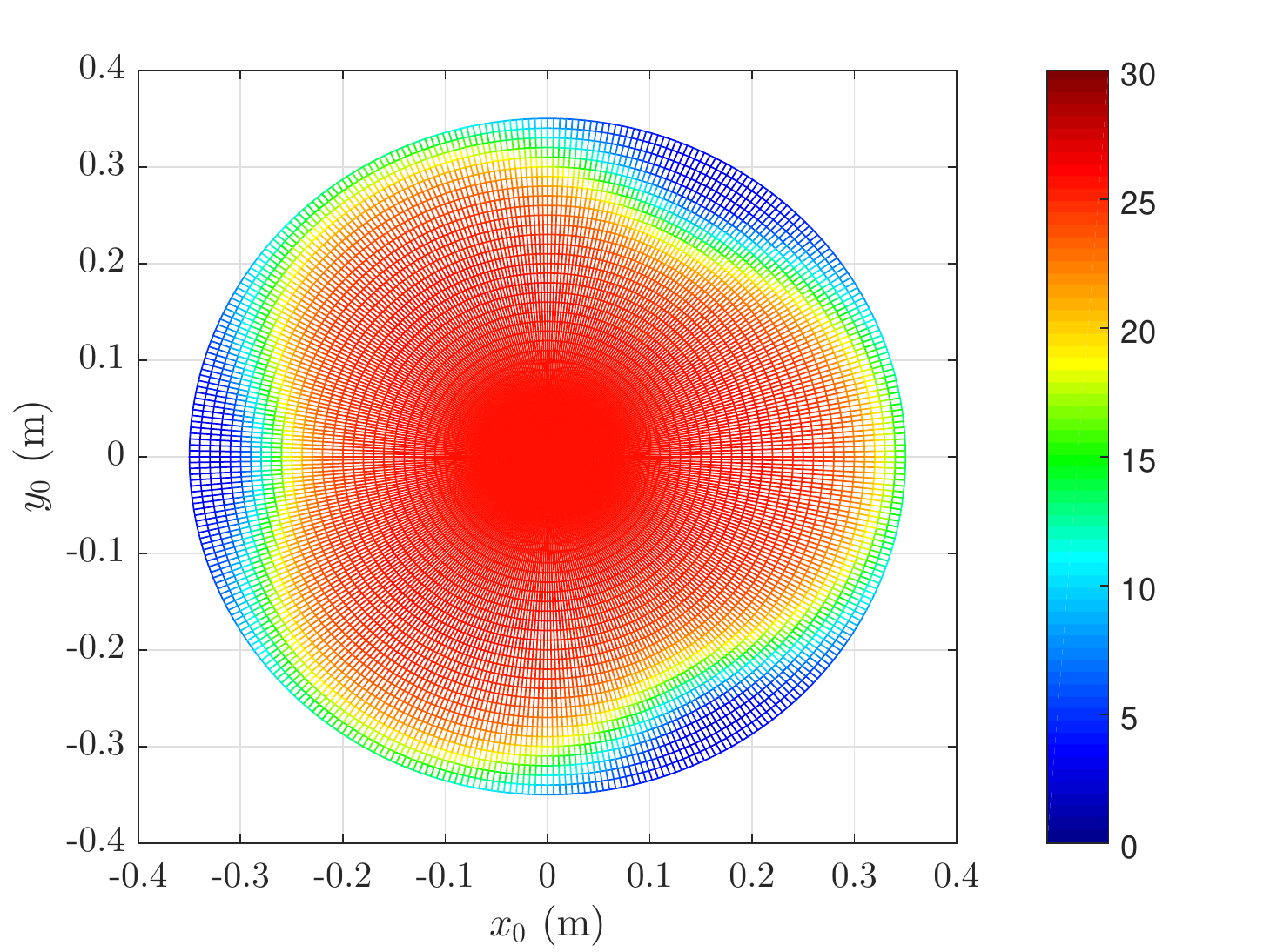}}  \hspace{-13mm}
		\scalebox{0.34}{\includegraphics{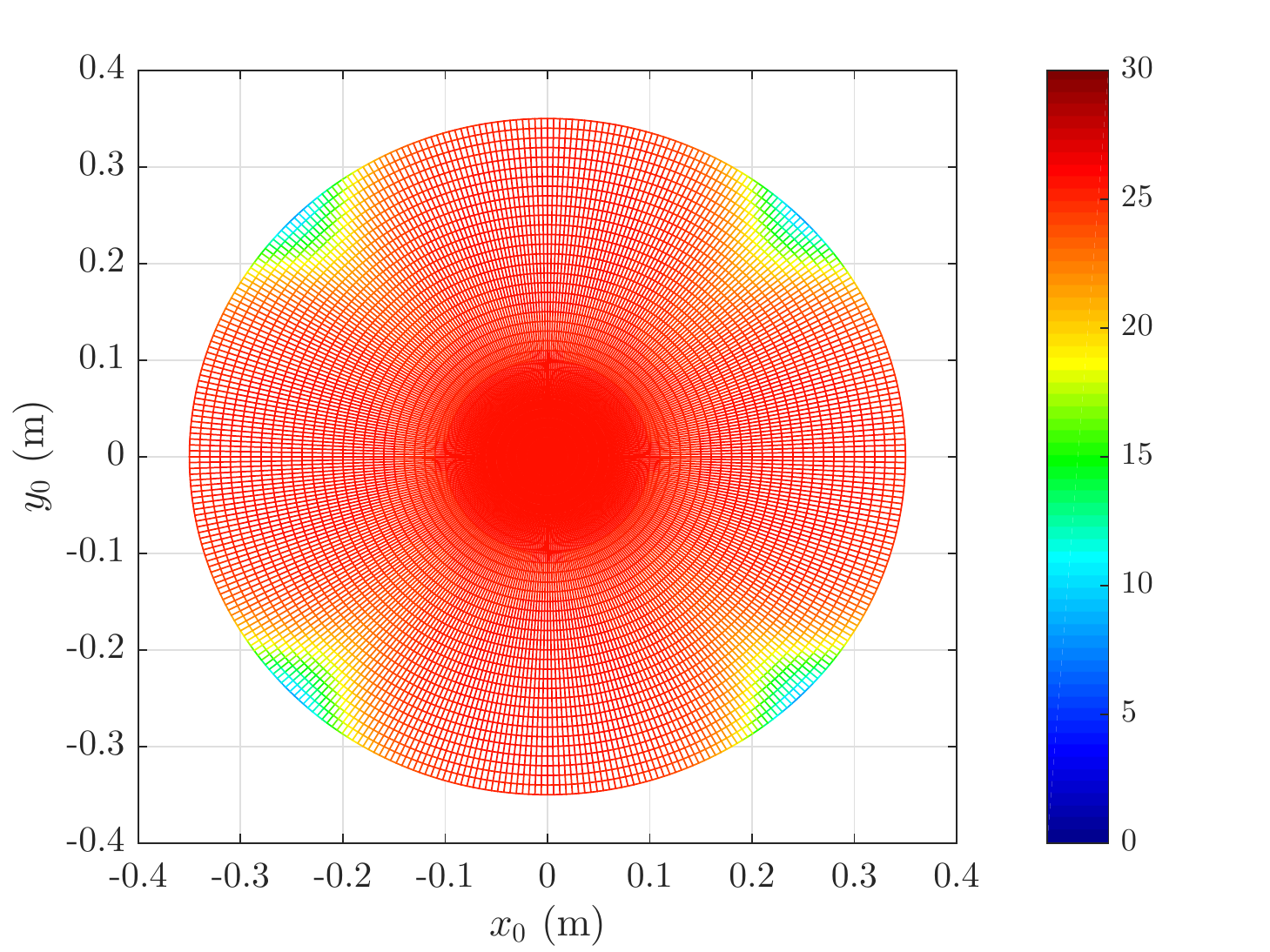}}}
		\subfigure[Centralized WPT] 
		{\scalebox{0.34}{\includegraphics{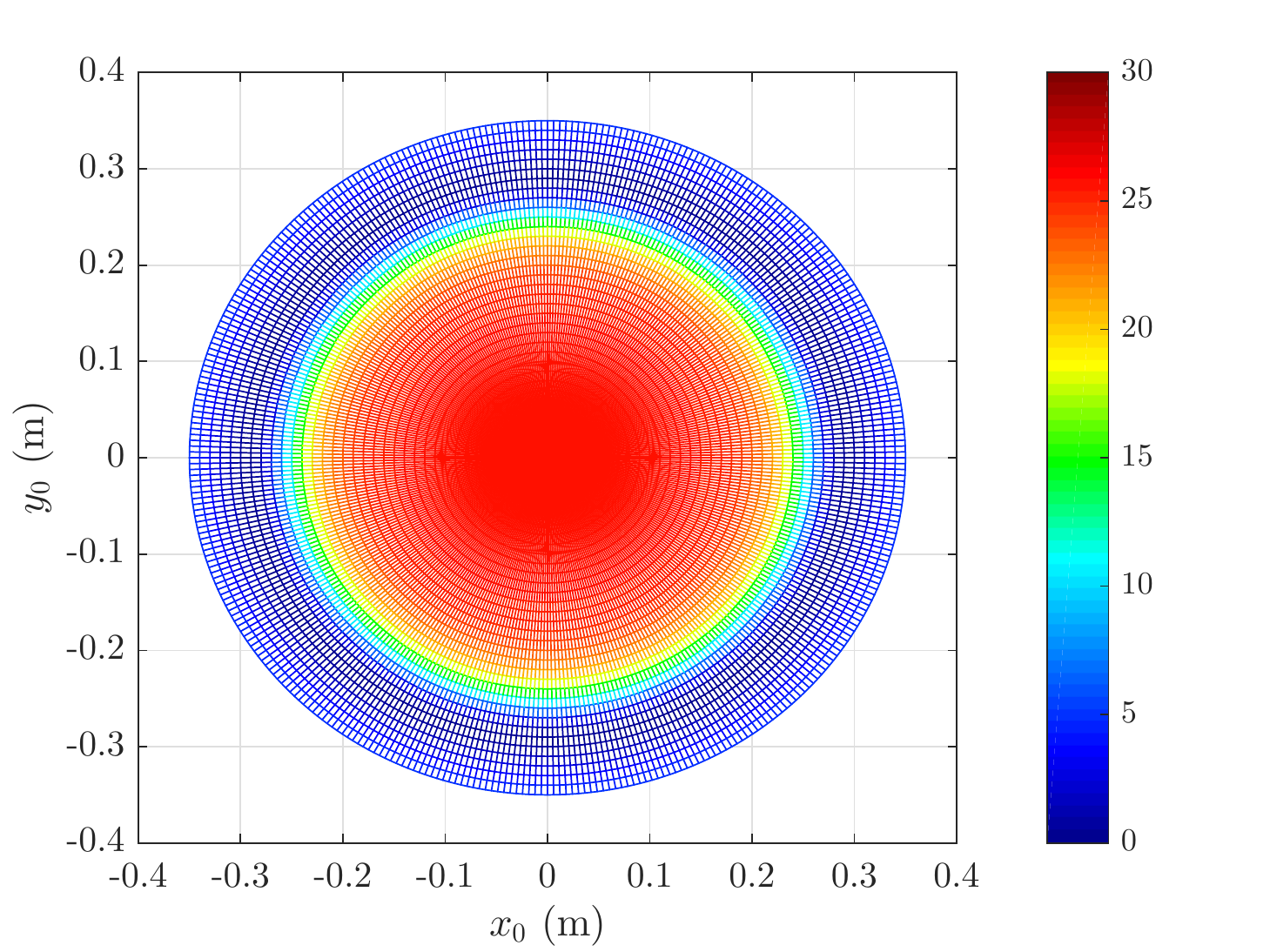}}}
	\end{center} 
	\caption{The load power distribution under different transmitter placement schemes in 2D.} \label{fig:2D-p0}
\end{figure*} 
%*************************
%*************************
%*************************
%*************************
\vspace{-1mm}
\subsection{Numerical Example} \label{Sec:2D-NumericalEcample}
To illustrate the performance of joint magnetic beamforming and transmitter location optimization in the 2D disk region case,  we consider the same system parameters as in Section \ref{Sec:numerical Example} for the 1D target line, which is now replaced by a disk  of radius $\rho=0.35$m (i.e., with $0.7$m in diameter which is the standard size for a round table with $2$--$3$ seats). Hence, the target region area ($0.385$m$^2$) is about ten times larger than the sum-area of all transmitter coils ($0.0393$m$^2$). 

As shown in Figs. \ref{fig:SymStructure}(a)--(c), three   rotationally symmetric structures exist for the system of $N=5$  transmitters. 
After obtaining the optimized transmitter locations for the three rotationally symmetric structures, we have  $\rho_1^\star=0.228$m  with  $\tau_1^\star=17.17$ for Structure $1$.   
For Structure $2$, we obtain  $\rho_1^\star=0.13$m  and $\tau_2^\star=2.85$.  For Structure $3$, we obtain $\rho_1^\star=\rho_2^\star=0.241$m,  $\phi_2^{\star}=\pi/2$,  and $\tau_3^\star=6.89$. 
Since $\tau_1^\star> \tau_3^\star> \tau_2^\star$, it follows that Structure $1$ has the best  performance in terms of maximizing the minimum deliverable power to the receiver load over the given target disk region $\cal R$, with the  system setup considered above.
For benchmark performance, we also consider centralize WPT (see Fig. \ref{fig:WPTTable}(a)) where  the five   transmitters are all placed at the origin $(x=0,y=0,z=0)$. Note that this benchmark structure is a special case of Structures $1$--$3$.  

Fig. \ref{fig:2D-p0} shows the power deliverable to the receiver load versus its location $(x_0,y_0)$ in $\cal R$, under distributed WPT including the three rotationally symmetric structures as well as centralized WPT (benchmark structure), with the optimized transmitter locations in each of the  three structures in distributed WPT and the optimal magnetic beamforming adaptive to the receiver location applied.  
The detailed performance  comparison among the four structures is  summarized in Table \ref{tab:quantitiveComparision3}, from which it is observed that the minimum deliverable power achieved by Structure $1$ is indeed much larger than those of  the other structures based on the actual simulation results.  
Note that $p_{\min}$ reported in Table \ref{tab:quantitiveComparision3} for  each of Structures $1$, $2$, and $3$ slightly differs  from  $\tau_s^{\star}$ obtained previously by solving its corresponding optimization problem. 
This is due to the fact that the approximation in (\ref{eq:hn0-simplified}) is used to compute $h_{n0}$'s in the node placement optimization problem, but the actual mutual inductance expression in (\ref{eq:h_n0}) is used in all simulations to achieve the best   accuracy.  
\begin{table}[t!]
	\centering
	\caption{Performance comparison between different transmitter placement schemes in 2D.}
	\begin{tabular}{|c|c|R|R|R|R|}
		\hline
		\multicolumn{2}{|c|}{Scheme} & $p_{0,\text{avg}}$ (W) & $p_{0,\min}$ (W) & $p_{0,\max}$ (W) & $\xi$ ~~~ (\%) \\ \hline
		\multirow{3}{*}{Distributed} & Structure $1$ & $24.02$ & $18.24$ & $25.54$ & $71.42$ \\ \cline{2-6} 
		& Structure $2$ & $22.31$ & $2.70$ & $25.62$ & $10.54$ \\ \cline{2-6} 
		& Structure $3$ & $24.64$ & $8.15$ & $25.54$ & $31.91$ \\ \hline
		\multicolumn{2}{|c|}{Centralized} & $17.93$ & $0$ & $25.65$ & $0$ \\ \hline
	\end{tabular}  \vspace{-3mm}
	\label{tab:quantitiveComparision3}
\end{table}
\begin{table}[t!]
	\centering
	\caption{Impact of region radius $\rho$ on the minimum deliverable power $p_{0,\min}$ under different transmitter placement schemes in 2D.} \vspace{-1mm}
	\begin{tabular}{|R|c|c|c|c|}
		\hline
		\multirow{3}{*}{\begin{tabular}[c]{@{}c@{}}$\rho$ (m)\end{tabular}} & \multicolumn{4}{c|}{$p_{0,\min}$ (W)} \\ \cline{2-5} 
		&  \multirow{2}{*}{Centralized} &\multicolumn{3}{c|}{Distributed} \\ \cline{3-5}
		&  &Structure $1$ & Structure $2$ & Structure $3$ \\ \hline
		$0.1$ & $25.59$ & $25.59$ & $25.59$ & $25.59$ \\ \hline
		$0.3$ & $0$ & $23.11$ & $4.33$ & $18.38$ \\ \hline
		$0.6$ & $0$ & $3.29$ & $2.69$ & $3.31$ \\ \hline
	\end{tabular}
	\label{tab:quantitiveComparision4} \vspace{-2mm}
\end{table} 

Next, Table \ref{tab:quantitiveComparision4} shows the impact of changing the target disk radius $\rho$ on the performance of  WPT in 2D.  
First, it is observed that when $\rho$ increases, the minimum deliverable power $p_{\min}$ decreases for all structures.   
It is also observed that when $\rho=0.1$m, the four structures perform the same, since the optimal solution is to place all the transmitters at the center $(x=0,y=0,z=0)$.  It is further observed that when $\rho=0.3$m, Structure $1$ outperforms the other structures, while Structure $3$ achieves slightly larger $p_{0,\min}$ over the other cases when $\rho=0.6$m. 
Last, it is observed that the minimum deliverable power of centralized WPT (benchmark structure) significantly drops when $\rho>0.1$m, which shows the inefficacy  of centralized WPT  in the 2D case. 
%*************************
%*************************
%*************************
\vspace{-2mm}
\section{Conclusion} \label{Sec:Conclusion}
In this paper, we study the node placement optimization for a MISO MRC-WPT system with distributed magnetic beamforming. 
First, we  propose the optimal magnetic beamforming solution to jointly assign the currents at different transmitters subject to their sum-power constraint with given locations of the  transmitters and  receiver. 
We show that although distributed WPT with optimal magnetic beamforming achieves better performance  than  centralized WPT, the resulting  load power profile still fluctuates  over a given target region considerably. 
This motivates us to  formulate a node placement problem to jointly optimize the transmitter locations with adaptive magnetic beamforming to maximize the minimum  power delivered to the load over a 1D line region. 
We propose an efficient algorithm for solving this problem based on bisection method and gradient-based search, which is shown by simulation to be able to improve the load power distribution significantly. 
Finally, we extend our design approach to the case of 2D disk region and show that  significant performance gain can also be achieved in this case.    
In this paper, for simplicity we assume identical transmitter coils of equal size, while the performance of WPT may be further improved if the sizes of transmitter coils can be optimized jointly with the transmitter locations, which is an interesting problem worthy of further investigation.     
Moreover, in this paper we assume  that the transmitters and receiver are placed in two parallel planes, and thus their coils all have the same orientation.  Reformulating and solving the node placement problem for the general scenarios where the coils of transmitters and receivers can have arbitrary locations/orientations is also interesting to investigate in future work.      
%*************************
%*************************
%*************************
\appendix 
\subsection{Proof of Lemma \ref{Lemma:mutualInductance}} 
In (\ref{eq:h_n0}), we can express  $J_1(e_{\text{coil},\text{tx}} u) J_1(e_{\text{coil},\text{rx}} u)=$ $ e_{\text{coil},\text{tx}}e_{\text{coil},\text{rx}}u^2/4$ $+ \sum_{m_1=1}^{\infty}\sum_{m_2=1}^{\infty}$  $(-1)^{m_1+m_2} \dot{J}_{m_1, m_2}(u)$, with $\dot{J}_{m_1, m_2}(u)$ $=(e_{\text{coil},\text{tx}}u/2)^{2m_1+1}(e_{\text{coil},\text{rx}}u/2)^{2m_2+1}/$ $(m_1!m_2!(m_1+1)! (m_2+1)!)$. 
Given $e_{\text{coil},\text{tx}}, e_{\text{coil},\text{rx}} \ll z_0$, we have  $\dot{J}_{m_1,m_2}(u) e^{-z_0 u}\approx 0$ over $u \ge 0$, since its maximum value over $u$ is  $ \beta_{m_1,m_2}  (e_{\text{coil},\text{tx}}/z_0)^{2m_1+1}(e_{\text{coil},\text{rx}}/z_0)^{2m_2+1}$, with $\beta_{m_1,m_2}=((m_1+m_2+1)/\exp(1))^{2(m_1+m_2+1)} /$ $(m_1!m_2!(m_1+1)! (m_2+1)!)$, which decreases to zero as  $(e_{\text{coil},\text{tx}}/z_0)^{2m_1+1}\rightarrow 0$ and $(e_{\text{coil},\text{tx}}/z_0)^{2m_1+1}\rightarrow 0$ for $m_1, m_2 \ge 1$. 
Hence, we can simplify (\ref{eq:h_n0}) as \vspace{-1mm}
\begin{align}
h_{n0}&\approx \beta\int_0^{\infty} J_0(d_{n0}u) u^2 e^{-z_0 u}du, \label{eq:hn0_int_simplified}
\end{align}
where $\beta=\mu \pi b_\text{tx} b_\text{rx} e_{\text{coil},\text{tx}}^2  e_{\text{coil},\text{rx}}^2/4$ is defied for convenience. 

Next, let ${\cal J}_{0,\gamma}(\psi)={\cal L} \{J_0(\gamma u)\}$, where $\gamma$ denotes a real number and  ${\cal L}\{\cdot\}$ represents the Laplace transformer. Specifically, we have \vspace{-1mm}
\begin{align}
{\cal J}_{0,\gamma}(\psi)=\int_{0}^{\infty} J_0(\gamma u) e^{-\psi u} du=\dfrac{1}{\sqrt{\gamma^2+\psi^2}}. \label{eq:LaplaceJ0}
\end{align}
It is known that for any real function $o(u)$,  with $O(\psi)$ denoting its Laplace transform, we have ${\cal L} \{u^n o(u)\}= (-1)^n \partial^n O(\psi)/\partial \psi^n$, $n=1,2,$ and so on. 

From  (\ref{eq:hn0_int_simplified}) and (\ref{eq:LaplaceJ0}), it then follows that $h_{n0}\approx \beta \partial^2 {\cal J}_{0,\gamma}(\psi)/ \partial \psi^2 = \beta (2\psi^2-\gamma^2)/(\gamma^2+\psi^2)^{5/2}$, with $\psi=z_0$ and $\gamma=d_{n0}$. The proof is  thus completed.
%*************************************
\vspace{-3mm}
\subsection{Proof of Proposition \ref{Proposition:Optimal_Current}}
For (P1), the optimal current solution   $i_n$'s to (P1) can be obtained by leveraging  the Karush-Kuhn-Tucker (KKT) conditions of the optimization problem  \cite{ConvOpt-Book}. Let $\lambda \ge 0$ denote the dual variable corresponding to the constraint (\ref{eq:p1_c1}).  The
Lagrangian of (P1) is given by
\begin{align}
&L= \dfrac{w^2 }{r_\text{rx}} \bigg(\dfrac{r_\text{rx,l}}{r_\text{rx}}-\lambda \bigg)\bigg(\sum_{n=1}^{N} h_{n0}\overline{i}_n\bigg)^2 \hspace{-2mm}-\lambda\bigg(r_\text{tx}\sum_{n=1}^{N} \overline{i}_n^2- p_{\text{max}} \bigg). \label{eq:lagrangina}
\end{align}
The KKT conditions of (P1) are also given by
\begin{align}
& r_\text{tx}\sum_{n=1}^{N} {\overline{i}_n}^2+\dfrac{w^2}{r_\text{rx}}  \bigg(\sum_{n=1}^{N} h_{n0}\overline{i}_n\bigg)^2\le  p_{\text{max}}, \label{eq:primalFes} \\
& \lambda\ge 0, \label{eq:dualFes}\\
&\dfrac{w^2 h_{n0}}{r_\text{rx}} \bigg(\dfrac{r_\text{rx,l}}{r_\text{rx}}-\lambda\bigg)\bigg(\sum_{k=1}^{N} h_{k0}\overline{i}_k\bigg) -\lambda  r_\text{tx}\overline{i}_n=0, ~ \forall n, \label{eq:KKT-Zero-Deriv} \\
&\lambda \bigg(r_\text{tx}\sum_{n=1}^{N} {\overline{i}_n}^2+\dfrac{w^2}{r_\text{rx}}  \bigg(\sum_{n=1}^{N} h_{n0}\overline{i}_n\bigg)^2 - p_{\text{max}} \bigg)=0. \label{eq:KKT-compl-slackn}
\end{align}
where (\ref{eq:primalFes}) and (\ref{eq:dualFes}) are the feasibility conditions for the primal and dual solutions,  respectively,  (\ref{eq:KKT-Zero-Deriv}) is due to the fact that the gradient of the Lagrangian with respect
to the optimal primal solution  $\overline{i}_n$'s must vanish, and (\ref{eq:KKT-compl-slackn}) stands for the complimentary
slackness. To solve the set of equations in  (\ref{eq:primalFes})--(\ref{eq:KKT-compl-slackn}), we consider two possible cases as follows.

$\bullet$ Case 1: $\lambda=0$. It can be  verified that any set of  $\overline{i}_n$'s satisfying $\sum_{n=1}^{N} h_{n0} \overline{i}_n=0$ and $r_{\text{tx}}\sum_{n=1}^{N} \overline{i}_n^2 \le p_{\text{max}}$ can  satisfy  the KKT conditions (\ref{eq:primalFes})--(\ref{eq:KKT-compl-slackn}) in this case. 
However,  the resulting $\overline{i}_n$'s will make the objective function of (P1) in (\ref{eq:p1_obj}) equal to zero, which cannot be the optimal value of (P1); therefore, this case cannot lead to the optimal solution to (P1).

$\bullet$ Case 2: $\lambda>0$. From (\ref{eq:KKT-Zero-Deriv}), it follows that $\overline{i}_k=(h_{k0}/h_{n0})\overline{i}_n$, $\forall k\neq n$. Moreover, from  (\ref{eq:KKT-compl-slackn}), it follows that $r_\text{tx}\sum_{n=1}^{N} {\overline{i}_n}^2+  (w\sum_{n=1}^{N} h_{n0}\overline{i}_n)^2/r_\text{rx}=p_{\text{max}}$. Accordingly, we obtain  $\overline{i}_n=\kappa  h_{n0}$, $n=1,\ldots,N$, and $\lambda=(w^2 r_{\text{rx,l}} \sum_{n=1}^{N} h_{n0}^2)/(r_{\text{tx}}r_{\text{rx}}^2+w^2 r_{\text{rx}}  \sum_{n=1}^{N} h_{n0}^2)$, where $\kappa$ is given by  \vspace{-3mm}
\begin{align}
\kappa=\dfrac{\sqrt{p_{\max}}}{\sqrt{  \bigg(\sum_{n=1}^{N} h_{n0}^2 \bigg) \bigg(r_\text{tx}+\dfrac{w^2}{r_\text{rx}} \sum_{n=1}^{N} h_{n0}^2 \bigg) }}. 
\end{align}
The obtained $\overline{i}_n$'s and $\lambda$  satisfy the KKT conditions  (\ref{eq:primalFes})--(\ref{eq:KKT-compl-slackn}). 

Note that except the above set of primal and dual solutions to (P1), $\overline{i}_n$'s and $\lambda$, given in Case  $2$, there is no other solution that  satisfies the KKT conditions  (\ref{eq:primalFes})--(\ref{eq:KKT-compl-slackn}).  
Thus, we can conclude that the solution obtained in Case $2$ is indeed the optimal solution to (P1) because the KKT conditions are necessary (albeit not necessarily sufficient) for the optimality of a non-convex optimization problem, which is the case of (P1).    
The proof is thus completed. 
%*************************************
\vspace{-3mm}
\subsection{Proof of Lemma \ref{Lemm:SymmCond}}
By definition,  \textit{rotational symmetry} refers to the property of an object if it looks the same after a certain turn around its center.
Based on this definition, a structure with $Q$ transmitter rings, shown in Fig. \ref{fig:Symmetry-Proof}, is rotationally symmetric {\it iff} there exists a common rotation angle $\dot{\phi}$, with $0< \dot{\phi} < 2\pi$, such that by setting $\phi_q+\dot{\phi} \rightarrow \phi_q$, $q=1,\ldots, Q$, the locations of the transmitters over $\cal R$ are invariant.  
Note that as mentioned in Section \ref{Sec:2D_ProbForm}, we have assumed that the transmitters are equally separated  over each ring to satisfy the rotational symmetry.   
In the following, we show  a necessary and sufficient condition for such $\dot{\phi}$ to exist.  

Consider any ring $q$. Since $N_q$ transmitters are equally spaced over ring $q$,  it can be verified that by setting   $\dot{\phi}=2\pi k_q/N_q$ in rad, with $k_q=1,\ldots,N_q-1$, the ring looks the same after the rotation.  
Without loss of generality, we  set $k_q=1$. 
By considering all $Q$ transmitter rings, we thus have  $\dot{\phi}=2\pi k/u$, with $k=1,\ldots, \min_{q\in\{1,\ldots,Q\}} N_q/u$, where $u$ is a common divider such that $N_q \mod u=0$,  $q=1,\ldots, Q$. 
In this case, the condition $0< \dot{\phi} < 2 \pi$ (which is necessary and sufficient for rotational symmetry) holds {\it iff}  there exists  at least one common divider $u$ no smaller than $2$, i.e., $u \ge 2$. 
The proof is thus completed.   
%**************************************
%**************************************
\vspace{-3mm}
\subsection{Proof of Proposition \ref{proposition:QN}}
First, we show that a structure with $Q$ transmitter  rings is rotationally symmetric and distinct over $\cal R$ {\it iff}  $N_1=\ldots=N_Q=u$ and $u \in \mathbb{P}_N$. Then, we obtain $S_N$ for  $N\ge 2$.  
 
Firstly, we show that if $N_1=\ldots=N_Q=u$ and $u \in \mathbb{P}_N$, the considered structure with $Q$ transmitter  rings is rotationally symmetric  and distinct over $\cal R$.
With given $u \in \mathbb{P}_N$, we have $u\ge 2$.  
Hence, from Lemma \ref{Lemm:SymmCond}, it follows that the structure is rotationally symmetric over $\cal R$, since $N_q \mod u=0$, $q=1,\ldots,Q$, and $u\ge 2$.    
Moreover,  since $u$ is a prime number and we have $N_1=\ldots=N_Q=u$, it can be easily verified that it is impossible to divide each individual ring into two or more concentric  rings each with smaller number of transmitters than the original ring while still preserving rotational symmetry of the structure.   
This implies  that the structure is indeed distinct. The proof of the `if' part is thus completed.    

Secondly, we show that if the considered structure with $Q$ transmitter  rings is rotationally symmetric and  distinct over $\cal R$, then $N_1=\ldots=N_Q=u$ and $u \in \mathbb{P}_N$ must hold.       
In this case, since the structure is assumed to be rotationally symmetric over $\cal R$, from Lemma \ref{Lemm:SymmCond},  it follows that a common divider $u \ge 2$ exists such that $N_q \mod u=0$, $q=1,\ldots, Q$.     
Hence, we can represent $N_q=a_q u$, where $a_q=N_q/u$ is a positive integer.  
Moreover, from the distinction  of the structure,  it follows that $a_q=1$, $q=1,\ldots,Q$; otherwise, each ring with $a_q>1$ can be divided into $a_q$ rings each with $u$ transmitters (while still  maintaining rotational symmetry), which contradicts with the initial assumption of distinct  structure and thus cannot be true.    
As a result,  we have $N_1=\ldots=N_Q=u$.   
Next, by contradiction we prove  that $u \in \mathbb{P}_N$  also  holds.  
Suppose  $u \notin \mathbb{P}_N$. 
Accordingly, we can rewrite  $u=a \dot{u}$, where $\dot{u} \in \mathbb{P}_N$ and $a=u/\dot{u}$ is a positive integer.    
In this case, each ring $q$ can be divided into $a$ rings each with $\dot{u}$ transmitters equally spaced over it, where  these rings can have different radii and  rotation angles in general, while in our case their radii are all set the same.    
This means that the considered structure is a special case of a more general structure with more rings each consisting of a prime number ($\dot{u}$) of transmitters, and thus cannot be distinct. 
Thus,  $u \in \mathbb{P}_N$  must hold. The proof of the `only if' part is thus completed. 
 
To sum up, we have shown that a  structure with $Q$  transmitter rings is rotationally symmetric and distinct {\it iff}  $N_1=\ldots=N_Q=u$ and $u \in \mathbb{P}_N$. With this result, it immediately follows that  $S_N=|\mathbb{P}_N|$. The proof of this proposition is thus completed.  
%*************************
%*************************
%*************************
%*************************

%*********************************
%*********************************
%*********************************

\end{document}